\documentclass[12pt]{article}
\usepackage{amsmath}
\usepackage{graphicx}
\usepackage{enumerate}
\usepackage{natbib}
\usepackage{url} 

\newcommand{\blind}{1}

\usepackage{titlesec}
\titlelabel{\thetitle.\quad}

\usepackage[utf8]{inputenc} 
\usepackage[T1]{fontenc}    
\RequirePackage[colorlinks,citecolor=blue,urlcolor=blue]{hyperref}
\usepackage{url}            
\usepackage{booktabs}       
\usepackage{amsfonts}       
\usepackage{nicefrac}       
\usepackage{microtype}      
\usepackage[dvipsnames]{xcolor}         

\usepackage[normalem]{ulem}

\usepackage{graphicx}
\graphicspath{ {./images/} }

\usepackage{csquotes}
\usepackage{amsmath}
\usepackage{amsthm}
\usepackage{amssymb}

\usepackage{algorithm}
\usepackage{algpseudocode}

\newtheorem{lemma}{Lemma}
\newtheorem{corollary}{Corollary}

\newtheorem{theorem}{Theorem}

\newtheorem{condition}{Condition}
\newtheorem{assumption}{Assumption}

\newcommand{\R}{\mathbb{R}}

\newcommand{\E}{\mathbb{E}}
\newcommand{\bP}{\mathbb{P}}

\newcommand{\cM}{\mathcal{M}}
\newcommand{\bM}{\mathbb{M}}

\newcommand{\Id}{\mathrm{Id}}

\newcommand{\pvalue}{\mathfrak{p}_{\rm value}}

\DeclareMathOperator*{\argmin}{arg\,min}

\newcommand{\iid}{\overset{\mathrm{iid}}{\sim}}

\usepackage{color}
\definecolor{stevecolor}{rgb}{0.86,0.08,0.24}
\definecolor{yuancolor}{rgb}{0,0.37,0.2}
\definecolor{yuancolor2}{rgb}{0,0.18,0.65}

\begin{document}

\def\spacingset#1{\renewcommand{\baselinestretch}%
{#1}\small\normalsize} \spacingset{1}


\if1\blind
{
  \title{\LARGE\bf Robust Distribution-Free Tests for the Linear Model}
  \author{Torey Hilbert, Steven MacEachern, Yuan Zhang\thanks{
        The authors gratefully acknowledge the NSF under grant numbers SES-1921523, DMS-2413823, and DMS-2311109
    }
    \hspace{.2cm}\\
    Department of Statistics, Ohio State University
    }
    \date{}
  \maketitle
} \fi

\if0\blind
{
  \bigskip
  \bigskip
  \bigskip
  \begin{center}
    {\LARGE\bf Robust Distribution-Free Tests for the Linear Model}
\end{center}
  \medskip
} \fi

\bigskip
\begin{abstract}
Recently, there has been growing concern about heavy-tailed and skewed noise in biological data. We introduce RobustPALMRT, a flexible permutation framework for testing the association of a covariate of interest adjusted for control covariates. RobustPALMRT controls type I error rate for finite-samples, even in the presence of heavy-tailed or skewed noise. The new framework expands the scope of state-of-the-art tests in three directions. First, our method applies to robust and quantile regressions, even with the necessary hyper-parameter tuning. Second, by separating model-fitting and model-evaluation, we discover that performance improves when using a robust loss function in the model-evaluation step, regardless of how the model is fit. Third, we allow fitting multiple models to detect specialized features of interest in a distribution. To demonstrate this, we introduce DispersionPALRMT, which tests for differences in dispersion between treatment and control groups. We establish theoretical guarantees, identify settings where our method has greater power than existing methods, and analyze existing immunological data on Long-COVID patients. Using RobustPALMRT, we unveil novel differences between Long-COVID patients and others even in the presence of highly skewed noise.
\end{abstract}

\noindent%

\vfill
\newpage
\spacingset{1.9} 
\section{Introduction: Long-COVID Immunological Features}\label{sec:intro-MY-LC}

Biological experiments producing high-throughput data yield massive data sets, but the data are often unclean: heavy-tailed, skewed, correlated, and generally irregularly distributed. These shortcomings of the data challenge type I error control in standard statistical analyses \citep{wang2015high, eklund2016cluster, sun2020adaptive}. Perhaps most dramatically, \citet{hawinkel2019broken} found that even commonly used nonparametric methods can suffer inflated false discovery rates in differential abundance testing. \citet{guan2023conformal} gives explicit examples of standard permutation tests failing to control type I error.

The immune profiles of Long-COVID patients from the Mount Sinai-Yale study (\emph{MY-LC} hereafter) \citet{klein2023distinguishing} show many of these shortcomings.  In \emph{MY-LC}, the proportions of immune cells of various types, $Y_i$, are compared between patients experiencing Long-COVID (LC) and healthy patients. The ``types'' of cells considered varies from very broad, such as the proportion of live cells that are natural killer cells, to very specific, such as the proportion of CD8+ T-cells that express IL6.  The proportions also vary with general demographic features such as patient age and BMI, and we aim to investigate differences between healthy and LC patients after adjusting for these demographics.
Following \citet{klein2023distinguishing} and \citet{guan2023conformal}, we write the model
\begin{align}\label{eq:MY-LC}
    Y_i = \beta I\{\text{LC}_i = 1\} + Z_i^T \theta + \epsilon_i,
\end{align}
where $Z_i^T = (1, \text{age}_i, \text{sex}_i, \text{BMI}_i, \text{age}_i\times\text{BMI}_i, \text{sex}_i \times \text{BMI}_i)$ is the vector of covariates for which we wish to adjust, $\beta$ and $\theta$ are parameters, and $\epsilon_i$ is a noise term.
We wish to test the null hypothesis that Long-COVID has no effect on the cell type proportion $Y$ after adjusting for the covariates $Z$.
More specifically, in model \eqref{eq:MY-LC}, we wish to test $H_0: \beta = 0$. We make no assumptions on $\theta$, and search for a test that is valid for any value of $\theta$.

Figure \ref{fig:MY-LC-plots}a shows a normal QQ plot from a fit of model \eqref{eq:MY-LC}, assuming $\epsilon_i$'s are i.i.d.\ normal, for the proportions $Y$ for one type of cell. We see that there is shows strong right-skewness, clearly violating the assumption that the $\epsilon_i$'s are normally distributed.
In figure \ref{fig:MY-LC-plots}b we plot the residuals for a fit from model \eqref{eq:MY-LC} for three different types of cells. We see that the residuals for control patients are mildly skewed, while the residuals for the Long-COVID patients are much more skewed and have greater dispersion. This suggests that Long-COVID impacts the shape and dispersion of the error distribution and that explicit modeling of these effects will be scientifically informative.

\citet{klein2023distinguishing} noted
the increase in variation in immune response in Long-COVID patients, but did not formally test for this increase.
We propose extending model \eqref{eq:MY-LC} to
\begin{align}\label{eq:MY-LC-flex}
    Y_i = Z_i^T \theta + f(LC_i, \epsilon_i).
\end{align}
In this extended model, we test $H_0: f \text{ depends on } LC_i$.
Our RobustPALMRT framework allows us to test for this difference in variation after adjusting for background demographic covariates while controlling the type I error rate, which we term DispersionPALRMT.

\begin{figure}[!ht]
    \centering
    \includegraphics[width=\linewidth]{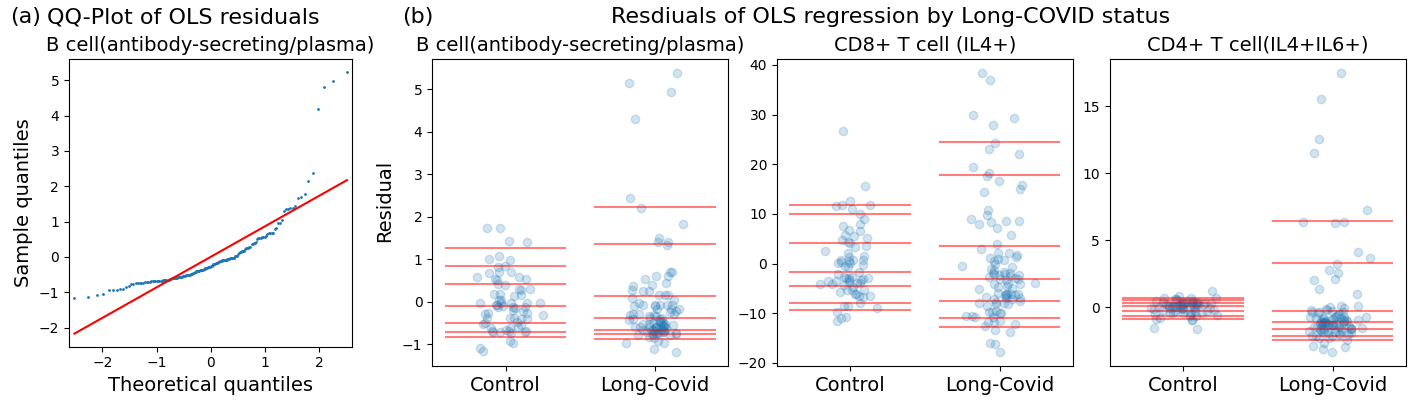}
    \vskip -0.25in
    \caption{\small{Residual analysis showing nonnormality and differences between Long-COVID and Control patient cell proportions under model~\eqref{eq:MY-LC} when fit via OLS. 
        Red lines in (b): $\{0.05, 0.10, 0.25, 0.50, 0.75, 0.90, 0.95\}$ quantiles.}
    }
    \label{fig:MY-LC-plots}
\end{figure}

\section{Prior work}

\subsection{Classical methods}

We consider a generalization of \eqref{eq:MY-LC} as the following linear regression problem
\begin{align}
    Y = X\beta + Z\theta + \epsilon,
    \label{regression-setup}
\end{align}
where $Y\in\mathbb{R}^{n\times 1}$ is the \emph{response}, 
$X \in \R^{n \times d}$ and $Z \in \R^{n\times p}$ are \emph{fixed covariates},
and $\epsilon\in\R^{n\times 1}$ is an error term whose entries are exchangeable.
Under model \eqref{regression-setup}, we aim to test $H_0: \beta = 0$.
Here we make no assumptions on $\theta$, and aim for tests that are valid for every $\theta$.

Many procedures exist to test $H_0$.
Until recently, these tests split into two groups.
The first group imposes strong parametric assumptions on $\epsilon$, such as  $\epsilon \sim \mbox{N}(0,\sigma^2 I)$ (e.g., \citet{fisher1922goodness}). 
The second group appeals to asymptotic arguments and large sample sizes, typically through a central limit theorem for some estimator of $\beta$ in \eqref{regression-setup}.
Later, the distributional assumptions on $\epsilon$ are relaxed, and one might, for example, use robust regression or quantile regression techniques (e.g., \citet{ronchetti2009robust,koenker1978regression}).
\citet{lei2021assumption} describe the ``century-long effort'' to develop tests for the partial correlation $\beta$, and provide an extensive literature review.
\citet{guan2023conformal} summarizes and critiques the finite-sample performance of many proposed permutation-based methods designed to yield an approximate type I error rate under weaker-than-normality conditions on the distribution of $\epsilon$. In general, thick-tailed distributions and outlier points can pose a problem for these classical methods.

\subsection{Finite-sample distribution-free tests}

\citet{lei2021assumption} developed the first test for $H_0$ that controls the type I error rate in finite samples under \eqref{regression-setup}.
Their Cyclic Permutation Test (CPT) relies on rank-constraints on the design matrix to construct a set of vectors $\eta_0, \ldots, \eta_m$ such that the linear statistics $(y^\top \eta_0, \ldots, y^\top \eta_m)$ are distributionally invariant to cyclic permutations.
However, the method requires $n/p$ to be at least 19 to perform a test at level $\alpha = 0.05$. Furthermore, the method does not have much higher power than the F-test in heavy tailed settings.

Other methods have been developed, including \citet{wen2022residual, guo2023invariance,  guan2023conformal, young2024asymptotically, d2024robust, spector2024mosaic} and \citet{pouliot2024exact}.
Many of these methods are tailored for specific variants of \eqref{regression-setup}.  For example, \citet{spector2024mosaic} aims to test the goodness of fit of a very specific model form and permits finite-sample heteroskedasticity according to that structure. 
\citet{guo2023invariance, young2024asymptotically} and \citet{d2024robust} 
all involve carefully constructed invariant substructures, similar to \citet{lei2021assumption}. While these methods are useful in testing for partial correlation, our focus is on moving beyond simple partial correlation to pick up (and understand) patterns such as those appearing in the \emph{MY-LC} data.

\citet{guan2023conformal} introduced a strikingly novel method called the \emph{Permutation Augmented Linear Regression Test (PALMRT)}.  
Previous permutation tests for the partial correlation assign a score to each permutation of the data and then compare the score for the observed data to this reference set.  
A typical score is the residual sum of squares from an OLS fit, perhaps accounting for the leverage of individual cases.  In contrast, PALMRT considers scores from pairs of permutations, establishes exchangeability of the rows of a certain matrix, and uses a technique from multisplit conformal prediction methods \citep{vovk2018cross, barber2021predictive, gupta2022nested} to ensure type I error control based on the comparison of an observed score to a set of permuted scores.  \citet{guan2023conformal} showed that PALMRT has up to $20$ times higher empirical power than CPT in some settings, and nearly the same power as the F-test in many settings.  

To the best of our knowledge, most existing works aim to match the F-test while adequately controlling for type I error in the difficult settings. However, if the analyst anticipates that the covariate and error distributions are thick-tailed, it is important to ensure that the method has high power in precisely those settings. Our work closes this gap and designs tests that have high power when the error distributions are thick-tailed or skewed.


\subsection{Our contributions}

We develop RobustPALMRT, a highly flexible framework that expands the scope of PALMRT well beyond least squares.  The extensions are designed with the motivating data example in mind (viz.\ the \emph{MY-LC} study).  They focus on settings where the error distribution may be far from normal -- precisely those settings where traditional tests have severely inflated type I error rates.  We make several specific contributions.  

First is expansion of the fitting procedure to any \emph{shift invariant} method, that is, any method that returns the same result for both $Y$ and $Y + Z\theta$ for any $\theta \in \R^p$.  These methods include robust regression, quantile regression, and many additional techniques.  Unlike least squares and quantile regression, most of these methods require specification (or estimation) of a tuning parameter, typically tied to the scale of the error distribution.  Estimation of the scale parameter must be handled carefully, as standard methods destroy type I error control.  Our framework allows for the estimation of the scale parameter in a fashion that maintains control of the error rate.  To our knowledge, this is the first test for robust regression parameters with scale estimation that is valid in a distribution-free setting without moment assumptions in finite-samples.  We illustrate the utility of this with a robust test for differences in center in the MY-LC study.

Second, we separate model fit from model evaluation. This allows the analyst to tailor their measure of fit to the alternative.  This allows one to fit the model efficiently under a solid set of assumptions while basing the test on a more robust summary of the residuals.  We illustrate the value of this split in a simulation study, where we find that the choice of the evaluation measure has a strong effect on power.  

Third, we enable the analyst to fit multiple models and combine the results for a single test.  This gives the analyst access to specific features of the residual distribution.  We illustrate this by developing \emph{DispersionPALRMT}, where we fit multiple quantile regression models to test for a difference in dispersion 
in the \emph{MY-LC} study.

\section{The General Framework}

To test the null hypothesis that $X$ has no effect on $Y$ after adjusting for $Z$, it suffices to test the adequacy of the following model:
\begin{align}\label{math:general-null}
    Y &\overset{H_0}{=} Z\theta + \epsilon, \qquad
    \epsilon_\pi \overset{d}{=} \epsilon \text{ for all $\pi \in S_n$},
\end{align}
where $S_n$ is the set of permutations of the integers $1, \ldots, n$, and where the notation $A_\pi$ denotes permuting the rows of matrix $A$ by permutation $\pi \in S_n$.  
Since the null model does not include $X$, if $X$ has an effect on $Y$ then $\epsilon$ will depend on $X$. Thus we wish to create a test that is sensitive to dependence between $X$ and $\epsilon$.

We briefly recall the core idea of PALMRT. One considers the pair of \emph{augmented} linear models $Y \sim X + Z + Z_\pi$ and $Y \sim X_\pi + Z + Z_\pi$, where the $Z_\pi$ augmentation allows for fair comparison of the effect of $X$ and $X_\pi$ on $Y$. 

Our RobustPALMRT framework has two components: a model-fitting algorithm $\cM$, and a model-evaluation procedure $\omega$.
The algorithm $\cM(Y, X, [Z, Z_\pi])$ produces a list of summary statistics describing how well the augmented model
\begin{align}
    Y = Z\theta + Z_\pi \theta^\prime + f(X, \epsilon)
\end{align}
fits the data, where $f$ is a function applied row wise to the $(X_i, \epsilon_i)$ pairs.
For example, if we choose to regress $Y \sim X + Z + Z_\pi$, the algorithm might output 
the vector of residuals $\cM(Y, X, [Z, Z_\pi]) = \mathbf{r} = Y - \hat{Y}$, the $L_2$ norm of the residual vector, or some other quantity of interest.
We then define the model evaluation function, $\omega : \bM \to \R$, with the convention that models that fit the data better have a smaller value of $\omega(M)$.
For an OLS model, $\omega(M)$ could be the sum of squared residuals.
We compute:
\begin{align} \label{M-pi}
    M^{\pi}_{\text{Orig}} =&~ \cM(Y, X, [Z, Z_{\pi}]), 
    \quad\textrm{and}\quad
    M^{\pi}_{\text{Perm}} = \cM(Y, X_{\pi}, [Z, Z_{\pi}]).
\end{align}
Intuitively, if $X$ has explanatory power for $Y$ after adjustment for $Z$, then for ``most'' permutations $\pi$, $M^{\pi}_{\text{Orig}}$ should be a better fit than $M^{\pi}_{\text{Perm}}$ \citep{guan2023conformal}.  To complete the test, select permutations, $\pi_1, \ldots, \pi_B \iid \text{Unif}(S_n)$, and compute:
\begin{align}\label{p-val}
    \pvalue &= \frac{1 + \sum_{b = 1}^B I\left[\omega\left(M^{\pi_b}_{\text{Orig}}\right) \geq \omega\left(M^{\pi_b}_{\text{Perm}} \right)\right]}{1 + B}.
\end{align}

\subsection{The core theorem}

Here we present the core theorem that underlies RobustPALMRT.
\begin{assumption}\label{assumption:exchangeable}
    $\epsilon$ is an exchangeable random vector.
\end{assumption}
\begin{assumption}\label{assumption:linearity}
    $Y$ is linear in $Z$, that is $Y = Z\theta + f(X, \epsilon)$.
\end{assumption}

In addition to these two assumptions, we require two conditions on the model fitting method.
Note that these conditions are explicitly verifiable.  The analyst can choose a procedure for which the conditions hold.   
\begin{condition}\label{cond:Z-shift}
    $\cM(Y + [Z, Z_{\pi}]\gamma, X, [Z, Z_{\pi}]) = \cM(Y, X, [Z, Z_{\pi}])$ for every $\gamma \in \R^{2p}$.
\end{condition}
\begin{condition}\label{cond:symm}
    $\cM(Y_\sigma, X_\sigma, [Z, Z_{\pi}]_\sigma) = \cM(Y, X, [Z, Z_{\pi}])$ for every permutation $\sigma \in S_n$.
\end{condition}
Condition \ref{cond:Z-shift} demands that the fitting algorithm depends only on the portion of $Y$ that is perpendicular to the $[Z, Z_\pi]$ subspace. For example, the residuals from any projection of $Y$ onto a subspace containing the columns of $[Z, Z_\pi]$ would satisfy this condition, along with many classical regression methods based on M-estimators.
Condition \ref{cond:symm} simply requires that the algorithm treats all cases symmetrically.


Notice that in assumption \ref{assumption:exchangeable} we do not require $\epsilon$ to be symmetric, which is critical because the MY-LC data exhibits high skewness.
We use assumption \ref{assumption:linearity} in conjunction with condition \ref{cond:Z-shift} to remove the effect of $Z$ by demanding that our fitting procedure be
invariant to a shift of $Y$ by $Z\theta$ for any $\theta \in \R^p$. 

\begin{theorem}\label{thm:validity}
    Suppose that $\cM$ satisfies conditions \ref{cond:Z-shift} and \ref{cond:symm}, and $Y$ satisfies assumptions \ref{assumption:exchangeable} and \ref{assumption:linearity}. Then for any $\alpha \in [0, 1]$ and $\theta \in \R^p$, we have
    \begin{align*}
        &\bP_{H_0}
        \left(\pvalue \leq \alpha \right) \leq 2\alpha.
    \end{align*}  
\end{theorem}
This theorem allows the analyst to strictly control the type I error rate by using a $pvalue$ cutoff of $\alpha / 2$ rather than $\alpha$.  However, the simulations in section \ref{sec:Simulations} suggest that for reasonable choices of $\cM$ the nominal type I error rate of $\alpha$ is attained without adjustment of the cutoff.  

\section{RobustPALMRT: Testing flexible hypotheses} \label{sec:robust-palmrt-tests}

\subsection{Finite-sample testing for M-estimators with scale estimation}\label{sec:palmrt-huber}

Robust regression via M-estimation is a standard method to ensure that an analysis is stable in the presence of skewed observations and outliers \citep{ronchetti2009robust, sun2020adaptive}.  To perform a robust regression, a criterion function, $\rho$, is chosen and the parameter estimates are found by empirical minimization of the criterion.  To perform a robust regression on the model (\ref{regression-setup}), compute $\big(\hat{\beta},\hat{\theta}\big) = \argmin_{\beta, \theta}  \sum_{i=1}^n \rho(r_i / \hat{s})$, where the vector of residuals is $\mathbf{r} = Y - X \beta - Z \theta$ and $\hat{s}$ is an estimated scale.  For a specific example, the criterion function may be Huber's linearization of quadratic loss, 
\begin{align*}
    \rho_{\text{Huber}}(t) = \begin{cases}
        \frac{t^2}{2} & |t| \leq \delta \\
        |t|\delta - \frac{\delta^2}{2} & |t| > \delta 
    \end{cases}
\end{align*}
where $\delta$ is traditionally set at $1.345$, with
scale estimation based on the median absolute deviation (MAD), computed by $\text{Median}(|r_i - \text{Median}(r)|)$.  

To create a test that fits within the RobustPALMRT framework, we need an algorithm $\cM$ that satisfies conditions \ref{cond:Z-shift} and \ref{cond:symm}.
The scale $\hat{s}$ is often concurrently estimated with $\big(\hat{\beta},\hat{\theta}\big)$ (e.g. \citet{ronchetti2009robust}, \citet{venables2013modern}) using an Iteratively Reweighted Least Squares (IRLS) procedure.
Informally, this algorithm proceeds as follows:
\begin{enumerate}
    \item Initialize a value $\big(\Tilde{\beta},\Tilde{\theta}\big)$ by OLS.
    \item Estimate the scale from the residuals in the model $\big(\Tilde{\beta},\Tilde{\theta}\big)$.
    \item Update $\big(\Tilde{\beta},\Tilde{\theta}\big)$ by weighted least squares with the fixed scale.
    \item Repeat steps 2 and 3 until convergence.
\end{enumerate}
In the supplemental materials algorithm A1, we provide the explicit algorithm implemented in the statistical software package R in the MASS library \citep{venables2013modern} and prove that it satisfies conditions \ref{cond:Z-shift} and \ref{cond:symm}.
To give a sketch of the proof, notice that every step in the algorithm is symmetric, so that condition \ref{cond:symm} is automatically satisfied. Then notice that after passing to the OLS residuals, every step depends on $Y$ only through the OLS residuals of $Y$ on $[X, Z, Z_\pi]$; since OLS satisfies the shift invariance condition, condition \ref{cond:Z-shift} is satisfied.

Our results show that one can choose $\cM$ to be the residuals of a Huber regression with MAD scale estimation and have type I error rate control for {\em any} comparison function $\omega$.  However, we are also concerned with power.  If the two regressions, $Y \sim X + Z + Z_\pi$ and $Y \sim X_\pi + Z + Z_\pi$, are fitted with separate scales, the resulting residuals are not comparable since they optimize different loss functions.
In our experience, it is more effective to estimate the scale from a preliminary regression, $Y \sim Z + Z_\pi$, and then use the same scale for both regressions.  The preliminary regression does not involve $X$ and is valid under $H_0$.  It tends to overestimate the scale under $H_a$, pushing more of the residuals towards the center region of the loss function.  This leads to a single scale estimate $\hat{s}$ from the preliminary regression and the order statistics of the residuals from a Huber regression with scale fixed at $\hat{s}$.  Denoting the ordered residuals as $\bf{r}$, the output of the algorithm is 
\[
\cM(Y, X, [Z, Z_\pi]) = (\hat{s}, \bf{r})
\]

\subsection{Model evaluation}

Model evaluation is distinct from model fitting.
In its basic form, RobustPALMRT compares many pairs of fitted models, as in \ref{M-pi}, identifying the member of each pair that fits the data better. When the summary of a model is a single vector of residuals, our focus has been on the norm of the residuals, and we might naturally consider $\omega(M) = \|\mathbf{r}\|$ for some norm $\|\cdot\|$.
PALMRT uses the $L_2$ norm for both model fitting (OLS) and model evaluation.
In keeping with our desire for robustness, we use the Huber loss\footnote{Notice that $\|\cdot\|_{\rm Huber}$ is \emph{not} a valid norm since it does not satisfy the triangle inequality. Nevertheless, we use notation reminiscent of a norm.}, 
$\|\mathbf{r}\|_{\text{Huber}} = \sum_{i = 1}^n \rho_{\text{Huber}}(r_i)$.

When using an estimated scale $\hat{s}$ from $\cM$ chosen as above, we should ensure that the models we are comparing are on the same scale before computing a Huber loss.
With the choice of $\cM$ as above, we use $\omega(M) = \left\| \mathbf{r} / \hat{s} \right\|_{\text{Huber}}$.
We describe our method precisely in the upcoming algorithm \ref{alg:mad-scale-full}.

\begin{algorithm}
\caption{Huber-Huber RobustPALMRT with MAD scaling}\label{alg:mad-scale-full}
\begin{algorithmic}[1]
    \Require Rectangular data set $[Y, X, Z]$. 
 $Y$ is the response vector; $X$ and $Z$ are matrices of covariates.  Tests for the effect of $X$, given the presence of $Z$ in the model.
        \For{$b = 1, \ldots, B$}
            \State Generate $\pi_b \sim \text{Unif}(S_n)$.
            \State Estimate $\hat{s}$ by Huber regression $Y \sim [Z, Z_{\pi_b}]$ with MAD scaling.
            \State Compute $\mathbf{r}^{\pi_b}_{\text{Orig}}$ by Huber regression $Y \sim [X, Z, Z_{\pi_b}]$ with fixed scale $\hat{s}$.
            \State Compute $\mathbf{r}^{\pi_b}_{\text{Perm}}$ by Huber regression $Y \sim [X_{\pi_b}, Z, Z_{\pi_b}]$ with fixed scale $\hat{s}$.
            \State $A_b \gets I\Big(
                \|\mathbf{r}^{\pi_b}_{\text{Orig}}/\hat{s}\|_{\text{Huber}} \geq
                \|\mathbf{r}^{\pi_b}_{\text{Perm}}/\hat{s}\|_{\text{Huber}}
            \Big)$ 
        \EndFor
        \State \Return $\pvalue \gets \frac{1 + \sum_{b = 1}^B A_b}{B + 1}$
\end{algorithmic}
\end{algorithm}

Our method provides a test for a specified value of $\beta$ under $H_0$.  By inverting the family of such $\alpha$-level tests, we obtain a confidence interval for $\beta$ with a guaranteed coverage probability of at least $100(1 - 2\alpha)$\%. 
As a technical note, the minimizer of the empirical loss may not be unique, as with the median when the sample size is even.  This issue is handled by specifying a convention that preserves condition \ref{cond:Z-shift}. Standard software packages do this.


\subsection{Finite-sample testing for multiple quantile regression fits}\label{sec:palmrt-quantile-regression}

Using quantile regression \citep{koenker1978regression}, RobustPALMRT can target complex features of the response distribution.
For example, with a pair of quantile regressions, one can estimate the conditional inter quantile range (IQR) to assess the dispersion of the response rather than the location.

In the MY-LC study, we are interested in identifying whether Long-COVID patients have increased variability in their cell type proportions controlled for demographic features.
Let $X$ be a vector of indicators for Long-COVID status, let $Z$ be the matrix of demographic features, and let $Y$ be the proportions of interest. 
If we fit two quantile regressions, one for the $90^{th}$ percentile and another for $10^{th}$ percentile, we can estimate the 80\% conditional IQR.
This is a robust measure of the dispersion, and may also be more scientifically interpretable.
For example, it is easier to interpret the implications of the statement ``For 80\% of the patients, the percent of CD8+ T-cells expressing IL6 is between 0.3\% and 11.4\%'' than it is to interpret the implications of a standard deviation of 7.8\% when the population is highly skewed.

\begin{algorithm}
\caption{DispersionPALRMT testing for differences in IQR}\label{alg:scale-palmrt}
\begin{algorithmic}[1]
\Require Rectangular data set $[Y, X, Z]$. 
 $Y \in \R^n$ is a vector of responses; $X \in \R^n$ is a vector of case/control indicators; $Z \in \R^{n \times p}$ is a matrix of covariates; quantiles $q_{\text{Low}}, q_{\text{High}}$.
\State Select $\pi_1, \ldots, \pi_B \iid \text{Unif}(S_n)$.
\For{$b = \{ 1, \ldots, B \}$}
    \For{$\tau = \{ \Id, \pi_b \}$}
        \State $Q^{\tau} ~~\gets$ Residuals of quantile regression of $Y$ on $[X_\tau, Z, Z_\pi]$, $q = \{ q_{\text{Low}}, q_{\text{High}} \}$
        \For{$j = \{ 0, 1 \}$}
            \State $\hat{s}_{j}^{\tau} \gets \frac{1}{|X_\tau = j|} \sum_{i \in \{ X_\tau = j \}} |Q_i^{\tau, q_{\text{High}}} - Q_i^{\tau, q_{\text{Low}}}|$
        \EndFor
        \State $r^{\tau} \gets -|\log\left({\hat{s}_{1}^{\tau}}/{\hat{s}_{0}^{\tau}}\right)|$
    \EndFor
    \State $A_b \gets I(r^{\Id} \geq r^{\pi_b})$
\EndFor
\State \Return $\pvalue \gets \frac{1 + \sum_{b = 1}^B A_b}{B + 1}$
\end{algorithmic}
\end{algorithm}

We state the \emph{DispersionPALRMT} method more formally for any case/control dataset. Let $n_0$ be the number of controls, and let $n_1$ be the number of cases. Define $q_{\text{Low}} = 0.10, q_{\text{High}} = 0.90$. We fit four quantile regressions: $Y \sim X + Z + Z\pi$ and $Y \sim X_\pi + Z + Z\pi$ at quantiles $q_{\text{Low}}$ and $q_{\text{High}}$. Define $\mathbf{r}^{q, k}$ to be the ordered residuals from quantile $q$ for cases in group $k \in \{\text{Control}, \text{Case}\}$. Then let
\begin{align*}
    \cM(Y, X, [Z, Z_\pi]) &= \big(
        \mathbf{r}^{q_\text{Low}, \text{Control}},~
        \mathbf{r}^{q_\text{Low}, \text{Control}},~
        \mathbf{r}^{q_\text{High}, \text{Case}},~
        \mathbf{r}^{q_\text{High}, \text{Case}}
    \big)   \\
    \omega(M) &= 
        -\left|\log\left(
            \frac
            {\frac{1}{n_1} \sum_{i = 0}^{n_1}
                \Big| \mathbf{r}^{q_\text{High}, \text{Case}}_i - \mathbf{r}^{q_\text{Low}, \text{Case}}_i \Big|}
            {\frac{1}{n_0} \sum_{i = 0}^{n_0}
                \Big| \mathbf{r}^{q_\text{High}, \text{Control}}_i - \mathbf{r}^{q_\text{Low}, \text{Control}}_i \Big|}
        \right)\right|.
\end{align*}


\section{Simulation studies}\label{sec:Simulations}

\subsection{Location simulations}

In our first set of simulations, we evaluate algorithm \ref{alg:mad-scale-full} in model \eqref{regression-setup}. In this setting, the F-test is the standard choice when the noise is normally distributed.
However, when the noise is substantially skewed, the F-test can have inflated type I error rates as we will see. We compare RobustPALMRT, using several choices of model fitting and model evaluation procedures, to the both PALMRT and the F-test.

We simulate data from model \eqref{regression-setup} for various choices of $X$, $Z$, $\epsilon$, and $\beta$.
We generate $\epsilon$ using (i) normal, (ii) $t_3$, (iii) Cauchy, (iv) multinomial + normal, and (v) log-normal distributions, respectively.
The multinomial + normal model, (iv), was proposed by \citet{guan2023conformal} as a challenging case for traditional permutation methods.
There, $n-1$ entries of $\epsilon$ are i.i.d.\ $N(0, 1)$, while one randomly selected entry has distribution $N(\pm 10^4, 1)$, representing a severe outlier.
There are $p = 6$ covariates in $Z$, while $X$ is univariate.
The entries of $[X, Z]$ are i.i.d.\ $\text{Cauchy}(0, 1)$, except for one intercept column in $Z$.
We use $\alpha = 0.05$ and $B = 999$ permutations for both PALMRT and RobustPALMRT. Lastly, we both test settings where $\beta = 0$, evaluating type I error rates, and where $\beta$ is selected such that the F-test has a specified power, thus evaluating power.

We consider two model fitting strategies: OLS and Huber robust regression; and three evaluation metrics: $L_1$, $L_2$ and Huber loss with an estimated scale parameter. 
We label the methods by fitting method and evaluation metric:  PALMRT from \citet{guan2023conformal} is OLS-L2 while algorithm \ref{alg:mad-scale-full} is Huber-Huber.

\begin{figure}
    \centering
    \includegraphics[width=1\linewidth]{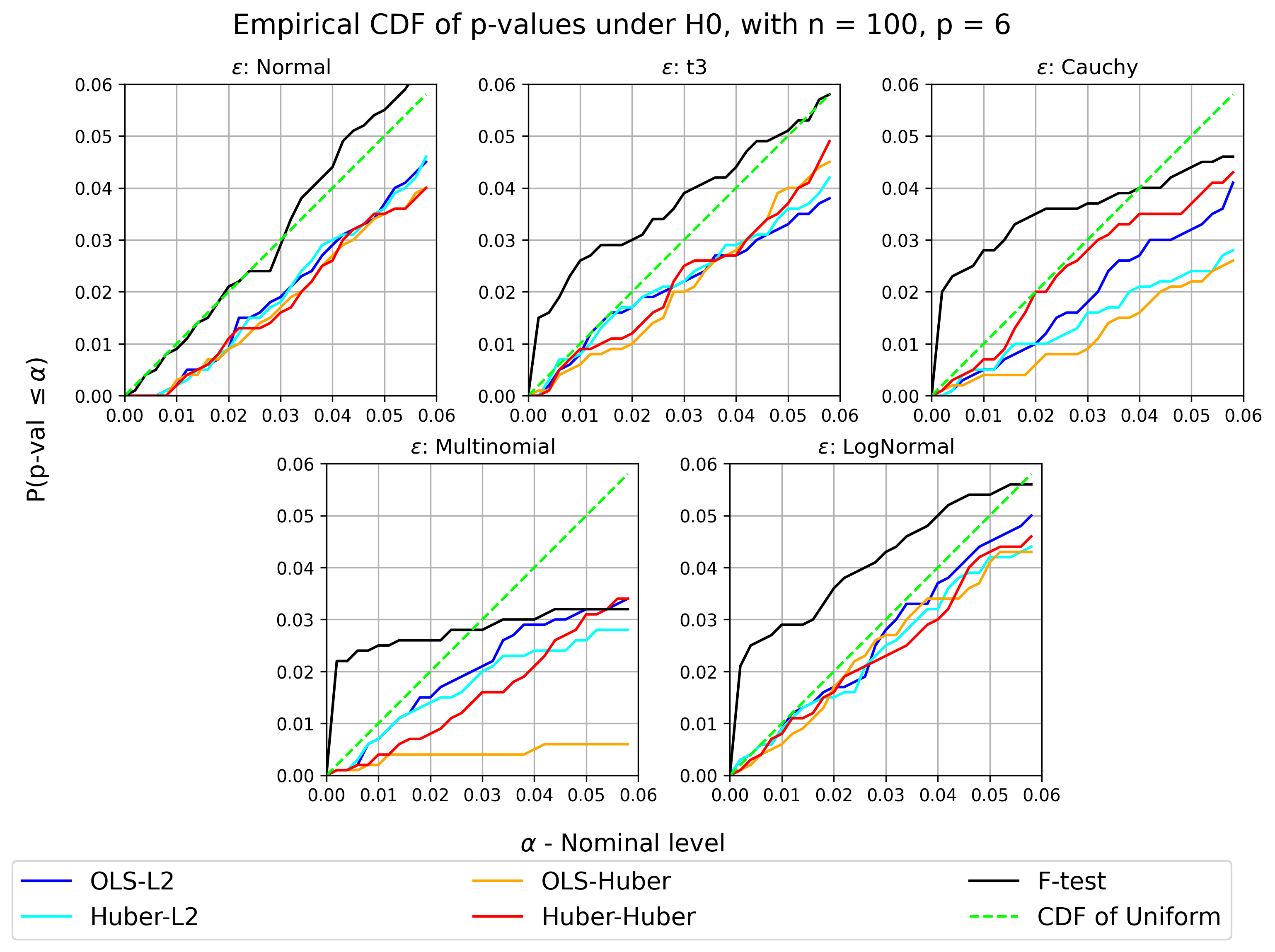}
    \vskip -0.25in
    \caption{
        \small{
        Empirical CDF of PALMRT and RobustPALMRT for p-values under $H_0$.
        If the p-values for a method are calibrated correctly, the simulated curve will lie near or below the dashed line, which is the cumulative distribution function (CDF) of the uniform distribution. Notice that for all nonnormal settings, the F-test spikes far above the uniform CDF for small nominal levels, while all RobustPALMRT methods stay below the uniform CDF.
        }
    }
    \label{fig:empirical_cdf}
\end{figure}

First we assess type I error with $\beta = 0$ with $n = 100$ cases. The F-test does well with the normal errors of setting (i).  It controls the type I error rate and is more powerful than the other methods, as expected.   
For the thick-tailed and skewed settings, the F-test has a greatly inflated type I error rate for small values of $\alpha$ and is far less powerful than the best of the other methods. See figure \ref{fig:empirical_cdf} for the empirical CDFs, and notice in particular that for very small nominal levels, say $\alpha < 0.01$, the F-test particularly struggles. Inflation of type I error at these low levels can negatively affect the false discovery rate control procedures commonly used with these types of datasets.

On the other hand, the type I error rate for both OLS-L2 (PALMRT) and Huber-Huber (RobustPALMRT) was found to be less than $\alpha$ in all tested settings.
Notably, theorem \ref{thm:validity} guarantees that both methods have type I error rate less than $2 \alpha$. The fact that these methods control type I error at the nominal level $\alpha$ in the simulations justifies our choice to refer to $\alpha$ as the nominal level despite the theory only guaranteeing control at level $2\alpha$. We are aware of extreme cases where the factor of $2$ is necessary, but the cases we know of require the use of model evaluation procedures that border on the absurd.

\begin{figure}
    \centering
    \includegraphics[width=1\linewidth]{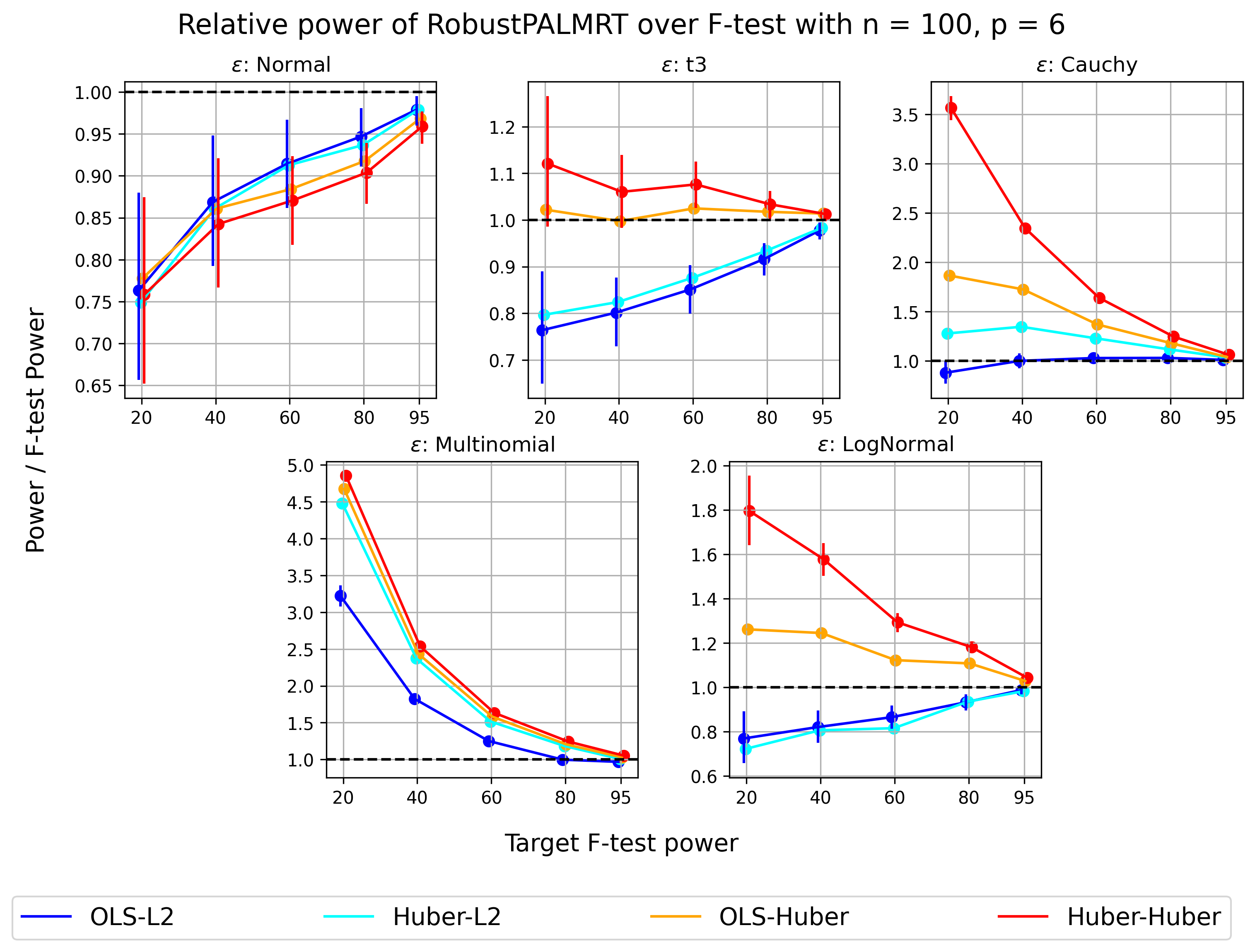}
    \vskip -0.25in
    \caption{
        \small{
        Relative power of RobustPALMRT to the the F-test. More specifically, (Power of RobustPALMRT) / (F-test power) vs F-test power.
        Error bars are provided for the OLS-L2 and Huber-Huber methods, but are very small.  The curves in these panels have been ``jittered'' horizontally to preserve their shapes while enhancing the visualization.  The nominal level of the F-test is $0.05$.
        Notice that for all nonnormal settings, the Huber-Huber method uniformly performed the best, followed by the OLS-Huber method.
        }
    }
    \label{fig:relative_sim}
\end{figure}

\begin{figure}
    \centering
    \includegraphics[width=1\linewidth]{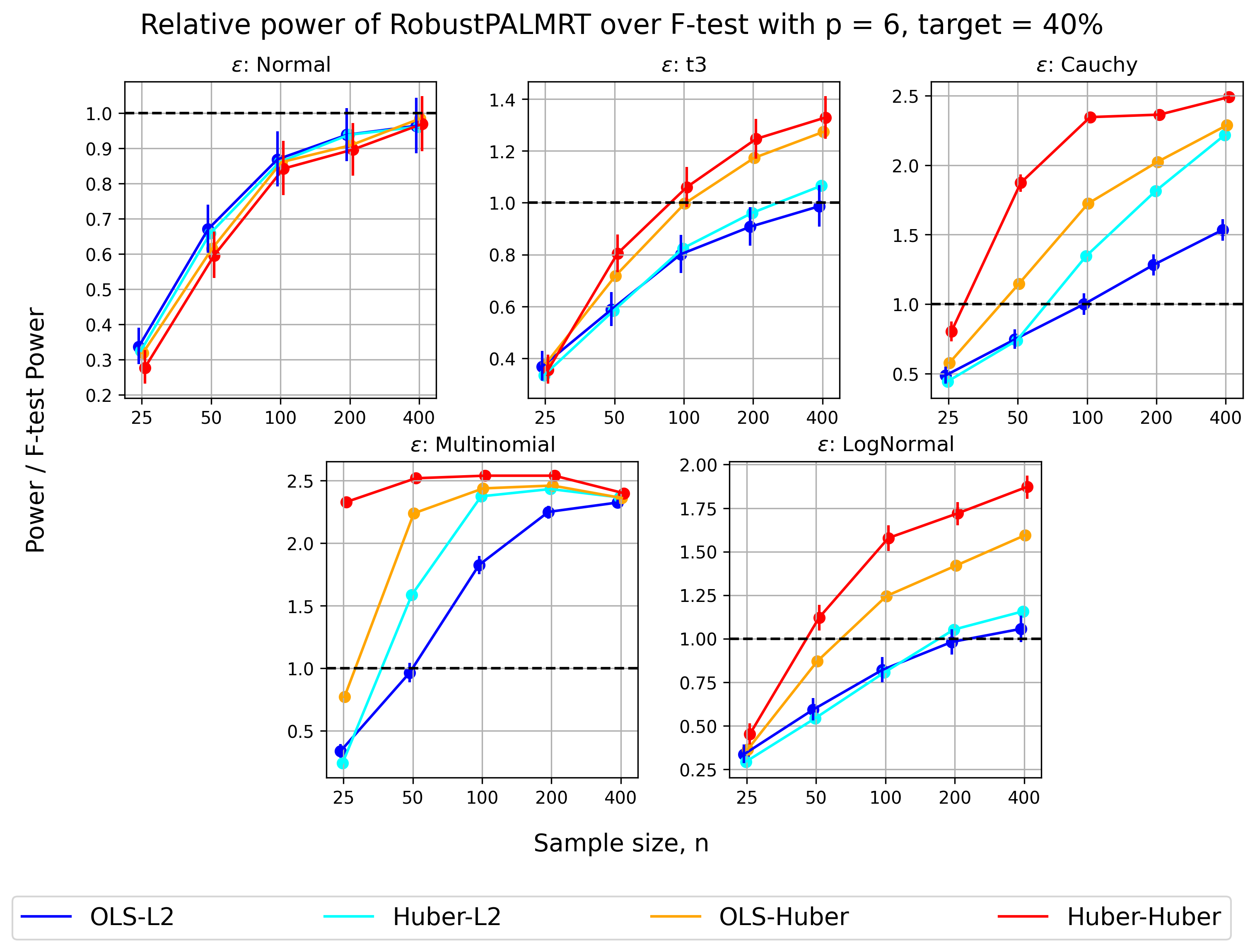}
    \vskip -0.25in
    \caption{
        \small{
            Relative power of RobustPALMRT to the the F-test. More specifically, (Power of RobustPALMRT) / (F-test power) vs sample size $n$, fixing F-test power at $40\%$.
            Error bars are provided for the OLS-L2 and Huber-Huber methods, but are very small.  The curves in these panels have been ``jittered'' horizontally to preserve their shapes while enhancing the visualization.  The nominal level of the F-test is $0.05$.
            Notice that for all nonnormal settings, the Huber-Huber method uniformly performed the best, followed by the OLS-Huber method.
        }
    }
    \label{fig:sample_size_sim}
\end{figure}

Next we assess the power under various alternatives.
We calibrate $\beta$ such that the F-test has an approximate power of $\{ 0.20, 0.40, 0.60, 0.80, 0.95 \}$.
In particular, we use a simulation-based root finding approach to select the $\beta$ in each of the $25$ settings such that the F-test will have the desired power. Calibrating $\beta$ in this way ensures that our tests cover a large range of power in each setting, and give a fair baseline for comparison with RobustPALMRT.

The OLS-L2 and Huber-Huber methods show similar performance when the errors are normal.
OLS-L2 has slightly greater power than Huber-Huber here. In figure \ref{fig:sample_size_sim}, we see that both methods converge toward the F-test's power as the sample size grows when the errors are normal.

For the thick-tailed and skewed error distributions of settings (ii) through (v), Huber-Huber has much more power than OLS-L2. In figure \ref{fig:relative_sim} we see that the Huber-Huber method uniformly outperformed all other methods, sometimes by a large margin. Interestingly, moving from OLS-L2 to OLS-Huber increases power substantially while preserving control of the type I error rate, suggesting that the model evaluation procedure plays a very important role.
This supports an emerging pattern in the literature on model choice where robust model evaluation leads to substantial improvement in power when the data is produced by a thick-tailed distribution (e.g., \citet{jung2021modified}).

Figure \ref{fig:sample_size_sim} presents results on sample sizes ranging from $n=25$ to $n=400$.  The patterns described above hold across this range of sample sizes.  We note that, for small sample size and nonnormal errors, the level of the F-test may be greatly inflated.  

PALMRT has been designed to provide a test with a (provable) bound on the finite-sample type I error rate when $\epsilon$ is not normally distributed.  It is in precisely these settings of nonnormal, but exchangeable, errors that RobustPALMRT outperforms PALMRT.  Our recommendation is to use Huber-Huber RobustPALMRT instead of PALMRT in these cases.  As a fallback position for those who wish to use OLS to fit the model, we recommend the use of a robust evaluation criterion (Huber loss) to increase the power of PALMRT while retaining the guaranteed type I error rate.  

\subsection{Dispersion simulations}

In our next set of simulations, we assess the performance of the quantile focused approach of section \ref{sec:palmrt-quantile-regression}, which we hereafter refer to as \emph{DispersionPALRMT}. We generate data from the following model:
\begin{align}\label{model:heteroskedasticity}
    Y_i &= Z_i^T \theta + (1 + \beta X_i) \epsilon_i,
\end{align}
where $\beta \in \R$ and $X \in \{0, 1\}^n$.
We generate $\epsilon$ according to (i) normal, (ii) Cauchy, (iii) log-normal distributions, respectively. We use $p - 1 = 5$ covariates in $Z$, with one intercept column and the rest drawn as $Z_{ij} \iid \text{Cauchy}$. To examine both type I error rate and power, we select $\beta \in \{ 0, 0.5, 1.0, 1.5, 2.0 \}$, and perform the hypothesis tests at level $\alpha = 0.05$.
Lastly, we select the sample size as $n \in \{100, 200, 400\}$.

\begin{figure}[htbp!]
    \centering
    \includegraphics[width=\linewidth]{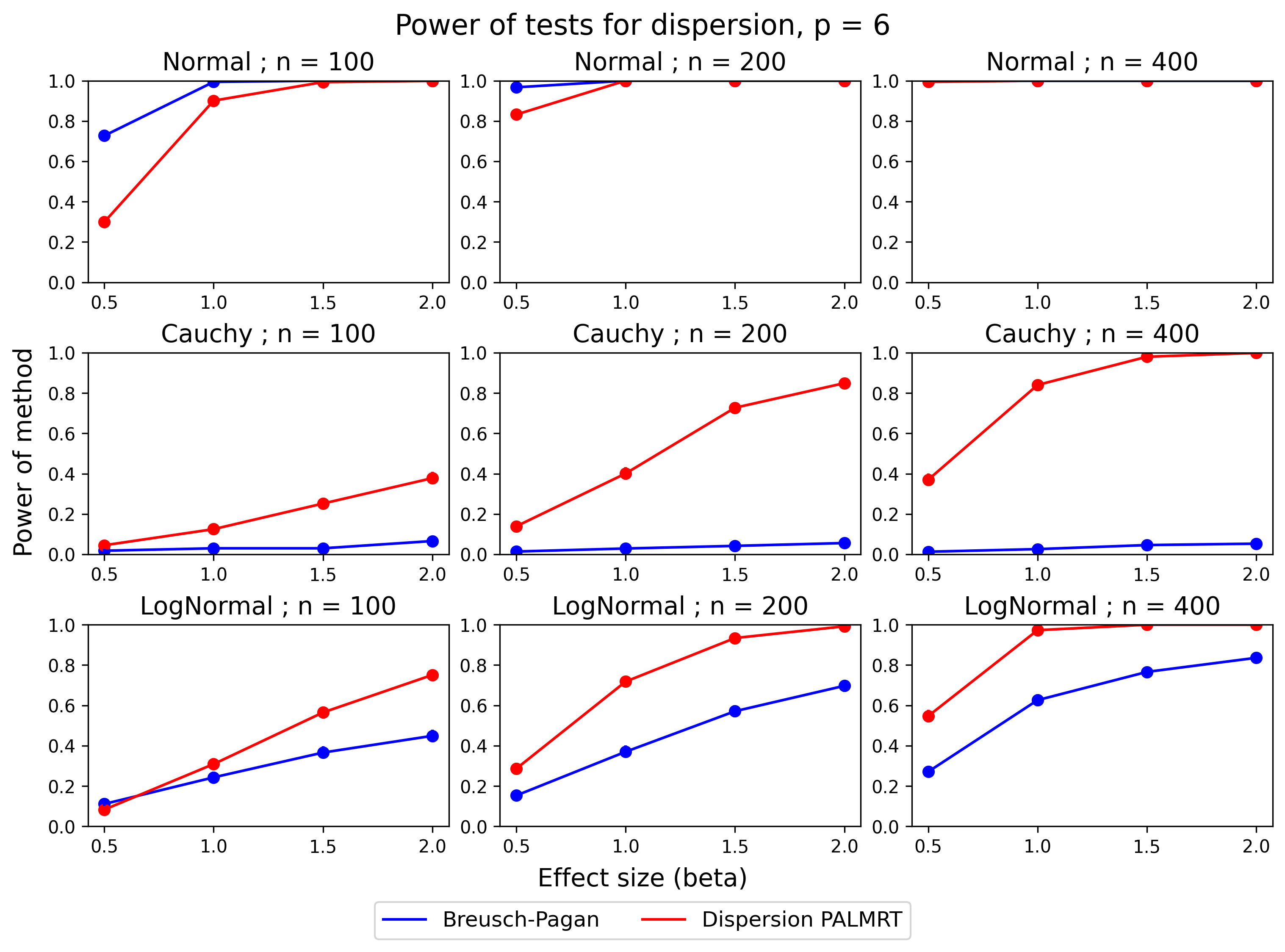}
    
    \caption{
        \small{Power of the Breusch-Pagan test and DispersionPALRMT in model \eqref{model:heteroskedasticity}. The power was computed by 1000 Monte-Carlo replicates at level $\alpha = 0.05$. Error bars are plotted, but are nearly invisible due to their negligible width. The title of each plot indicates the noise distribution and sample size. Notice that in the nonnormal settings, DispersionPALRMT (algorithm \ref{alg:scale-palmrt}) requires around $n = 200$ observations to detect a two-fold increase in standard deviation of the noise (i.e. $\beta = 1.0$).}
    }
    \label{fig:dispersion}
\end{figure}

We compare DispersionPALRMT to the Breusch-Pagan test \citep{breusch1979simple, cook1983diagnostics}, a classic test for heteroskedasticity such as that in model \eqref{model:heteroskedasticity}. For robustness in heavy tailed settings, we use the correction to the Breusch-Pagan test described in \citet{koenker1981note}, using the implementation in the statsmodels Python package \citep{seabold2010statsmodels}. This test rests on either normality assumptions or asymptotic theory, and so we expect that the test will perform well in the normal noise of setting (i) but suffer in the nonnormal noise of settings (ii) and (iii).

All methods controlled the type I error rate in all tested settings, so we omit plots of performance under the null. We note, however, that in all tested settings DispersionPALRMT was highly conservative. For example, it has a type I error rate of $0.004$ in the normal errors setting with $n = 100$. This suggests that there is room to recalibrate the test to the nominal level. The Breusch-Pagan test was similarly conservative in the Cauchy errors setting, likely owing to the studentization correction of \citet{koenker1981note} dominating the signal. 

The power curves for the two method are shown in figure \ref{fig:dispersion}. In the normal setting (i), we see that the Breusch-Pagan test outperforms the dispersion focused PALMRT method. This is expected, since purely quantile based methods are generally less efficient than least squares based approaches in normal settings. 

In the nonnormal settings (ii) and (iii), the quantile based approach of DispersionPALRMT is much more effective than the Breusch-Pagan test. We attribute this to the fact that the quantile regressions are unaffected by the extreme observations that the heavy tailed errors create.

\section{Case Study: Long-COVID Immunological Features}\label{sec:MY-LC}

We use the RobustPALMRT framework to analyze the MY-LC study Long-COVID dataset presented in section \ref{sec:intro-MY-LC}.
This dataset was also analyzed \citet{guan2023conformal} using PALMRT, and so we compare the conclusions of the analyses between OLS-L2 PALMRT and Huber-Huber RobustPALMRT with scale estimation.

Seven patients were excluded from the study for biological reasons \citep{klein2023distinguishing}, leaving $99$ Long-COVID patients and $77$ control patients.
We analyze the proportions of $65$ different cell types, with each type modeled separately. The proportion of each cell type will serve as a response variable, say $Y^{(k)}$, $k = 1, \ldots, 65$. For each response, we use model \ref{eq:MY-LC}.

\begin{figure}[htbp!]
    \centering
    \includegraphics[width=\linewidth]{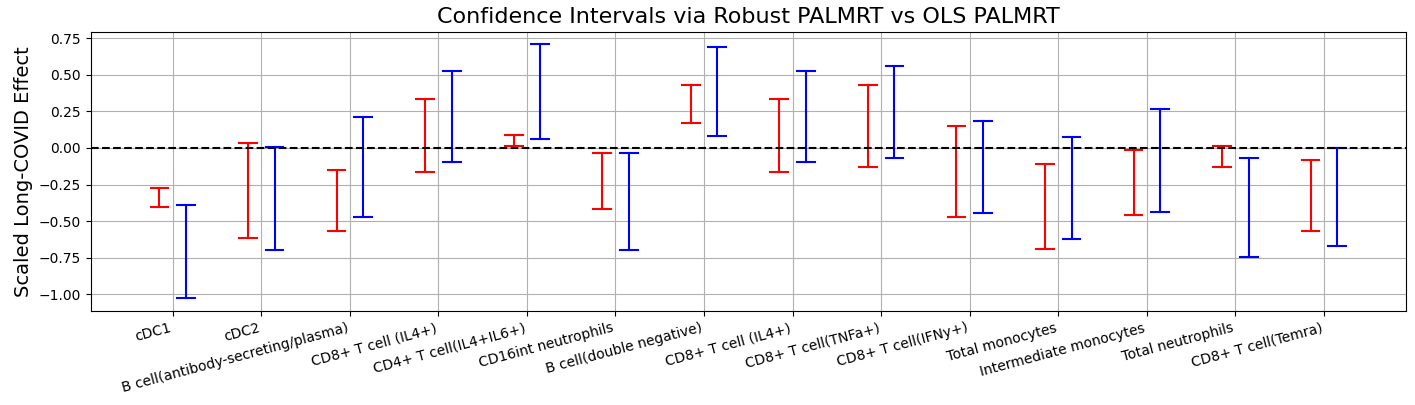}
    
    \caption{
        \small{95\% CI's for the effect of Long-COVID using Huber-Huber (red, left) and OLS-L2 (blue, right). The intervals are scaled by the standard deviation of the cell type proportion ($Y^{(k)})$.}
    }
    \label{fig:MY-LC-plots-CIs}
\end{figure}

Figure \ref{fig:MY-LC-plots-CIs} shows confidence intervals for the Long-COVID effect using OLS-L2 PALMRT and Huber-Huber RobustPALMRT for $14$ selected cell types.
Notice that for several cell types, one interval covers zero while the other does not, implying that the two tests result in different conclusions.
For many cell types the intervals are similar, with the RobustPALMRT interval tending to be a little shorter than the corresponding PALMRT interval.
However for some cell types there is a dramatic difference in interval width.
In these cases, we also find that the centers of the intervals are quite different.
We attribute these differences to the lack of robustness of OLS-L2 PALMRT.
The focus on the mean and variance that is implicit in the procedure produces a test (and interval) that is sensitive to extreme observations.
Huber-Huber RobustPALMRT reduces this sensitivity.  

Figure \ref{fig:MY-LC-plots} suggests that, for some cell types, the most evident impact of Long-COVID will be on the dispersion and skewness of cell type proportion.
We use the quantile regression approach described in section \ref{sec:palmrt-quantile-regression} to test for differences in the 80\% IQR between the Long-COVID and healthy control patients adjusted for control covariates.
Our dispersion-focused analysis found a significant effect of Long-COVID at the $\alpha = 0.05$ level for $15$ of the $65$ cell types.
The significant effects include those for which there was no apparent impact of Long-COVID on mean cell concentration; in particular for 3 cell types, the F-test did not find a significant difference in mean cell concentration while our dispersion-focused method found strong evidence (p-val $< 0.01$) of a difference in dispersion of cell concentration between the groups. Specifically, we found strong differences in the dispersion of the cell concentrations between Long-COVID and control patients in the following subpopulations: TNF$\alpha$-expressing T-cells among all CD8+ T-cells, IL4-expressing T-cells among all CD8+ T-cells, and IFN$\gamma$-expressing T-cells among all CD4+ T-cells.

\section{Proof Sketch of theorem 1}\label{sec:proof-thm-valid}

To prove the validity of theorem \ref{thm:validity}, we build off of the work of \citet{guan2023conformal}. We first give a sketch of that work, and then we show how our framework fits into that work. In the supplemental materials we give a full proof of an extended ``population version'' of the result that includes weights.

\citet{guan2023conformal} introduced the idea of forming an array of direct comparisons of permuted and unpermuted versions of the model rather than forming a ``one-dimensional'' list of summary statistics $T(\pi_1) \leq \cdots \leq T(\pi_B)$. Let $\pi_1, \pi_2 \in S_n$ be an arbitrarily selected pair of permutations. We construct an array $T(\pi_1, \pi_2; \epsilon)$ with a very specific symmetry property, namely that $T(\sigma \circ \pi_1, \sigma \circ \pi_2; \epsilon_\sigma) = T(\pi_1, \pi_2; \epsilon)$.

From here, we construct the $B \times B$ comparison matrix $A_{\pi_1, \pi_2} = I(T(\pi_1, \pi_2; \epsilon) \geq T(\pi_2, \pi_1; \epsilon))$. This matrix is stochastic, but appealing to a classic result from tournament theory \citep{landau1953dominance, barber2021predictive} one can deterministically show that there are at most $2\alpha B$ columns $j \in \{ 1, \ldots, B \}$ of $A_{i, j}$ such that $\frac{1}{B} \sum_{b = 1}^B A_{\pi_b, \pi_j} \leq \alpha$. By the exchangeability of $\epsilon$ and the aforementioned symmetry property, the ``identity column'' $\pvalue = \frac{1}{B}\sum_{b = 1}^B A_{\pi_b, Id}$ is exchangeable with the other $B - 1$ columns, and so $\bP(\pvalue \leq \alpha) \leq 2\alpha$.

Two things remain to be done:
\begin{itemize}
    \item Construct an array of comparisons $T(\pi_1, \pi_2; \epsilon)$ with the necessary symmetry property.
    \item Show that the constructed array of comparisons actually describes the p-value in \eqref{thm:validity}.
\end{itemize}
We construct the following array of comparisons between permutations $\pi_1, \pi_2 \in S_n$:
\begin{align*}
    T(\pi_1, \pi_2; \epsilon) &= \omega\big(\cM(\epsilon, X_{\pi_2}, [Z_{\pi_1}, Z_{\pi_2}])\big).
\end{align*}
Then we have, for any $\sigma \in S_n$,
\begin{align*}
    T(\pi_1, \pi_2; \epsilon_\sigma) &= \omega\big(\cM(\epsilon_\sigma, X_{\pi_2}, [Z_{\pi_1}, Z_{\pi_2}])\big) \\
    &= \omega\big(\cM(\epsilon, X_{\sigma^{-1} \circ \pi_2}, [Z_{\sigma^{-1} \circ \pi_1}, Z_{\sigma^{-1} \circ \pi_2}])\big) \\
    &= T(\sigma^{-1} \circ \pi_1, \sigma^{-1} \circ \pi_2; \epsilon).
\end{align*}
Then by theorem 3.1 in \citet{guan2023conformal}, under the null model we have
\begin{align}\label{math:eps-pval}
    \bP \bigg(\frac{1 + \sum_{b = 1}^B  I(T(\pi_b, Id; \epsilon) \geq T(Id, \pi_b; \epsilon))}{1 + B} \leq \alpha \bigg) \leq 2\alpha.
\end{align}
Notice that, for any $\theta \in \R^p$ and $\pi \in S_n$, we have
\begin{align*}
    T(\pi, Id; Y) &= \omega\big(\cM(Y, X, [Z_{\pi}, Z])\big) 
        = \omega\big(\cM(\epsilon + Z\theta, X, [Z_{\pi}, Z])\big) 
        = \omega\big(\cM(\epsilon, X, [Z_{\pi}, Z])\big) \\
        &= T(\pi, Id; \epsilon),\\
    T(Id, \pi; Y) &= \omega\big(\cM(Y, X_\pi, [Z, Z_{\pi}])\big) 
        = \omega\big(\cM(\epsilon + Z\theta, X_\pi, [Z, Z_{\pi}])\big) 
        = \omega\big(\cM(\epsilon, X_\pi, [Z, Z_{\pi}])\big) \\
        &= T(Id, \pi; \epsilon).
\end{align*}
Hence \eqref{math:eps-pval} implies $\bP(\pvalue \leq \alpha) \leq 2\alpha$, as desired.

\section{Discussion and Extensions}\label{sec:discussion}

We have established that RobustPALMRT controls the type I error rate for a test of model \eqref{math:general-null} under only exchangeability of $\epsilon$ and linearity of $Z$.
The introduction of the shift invariance condition to replace least squares projection expands the capabilities of \citet{guan2023conformal} in directions that align with sound data analysis.
In particular, we show that it is possible to perform scale estimation while maintaining shift invariance, and so our method can be applied to settings that are not scale invariant, something that commonly arises when conducting a data analysis.

The extension covers a broad swath of robust regression methods.
Our simulations show that the common pattern in robust regression holds in this setting -- little is lost when stringent parametric assumptions hold and much is gained when they do not hold.
In particular, we found that RobustPALMRT dramatically outperforms competing methods when there is strong skewness or heavy tails, while the method is nearly as efficient as the F-test when the errors are normal.

We have also demonstrated the flexibility of the framework by designing a quantile-based test for heteroskedasticity in immune responses of Long-COVID patients. Our simulations suggest that our test is more efficient than existing tests when the response has varying dispersion, a property the MY-LC dataset has.

While we have only demonstrated two possible uses for this framework, there are a vast number of potential extensions.
For example, we dealt exclusively with univariate responses, but the theory also holds for multivariate responses.
In immunology, it is common to look at sets of markers, and this can easily fit within RobustPALMRT. Future work in this direction might look into ways to take advantage of classical multiple comparison methods to search through these sets of markers.

Furthermore, the model fitting procedure could return an entire family of solutions, such as the entire model fitting path when varying the regularization parameter $\lambda$ in a LASSO fit.
In particular, this means that one can use a high dimensional set of features for $X$. Unfortunately, $Z$ is currently limited to being of modest dimension; we do know how to include high dimensional $Z$ given known bounds on the norm of $\theta$, but the method has low power, and so we did not include it in this work.

While powerful, the RobustPALMRT framework is not a panacea.
The theoretical results rest on the assumption that all nonlinearity is captured in the covariates of interest and the error, namely model \eqref{math:general-null}.
Exchangeability is an essential part of our current proofs, and our own preliminary investigation (not shown here) suggests that the type I error rate may not be controlled when this assumption is violated, for example, when there is moderately strong unmodeled heteroskedasticity.

Finally, given the scope of the framework, RobustPALMRT is not narrowly proscriptive. 
Analysts have the flexibility to select from many possible robust regression methods and many possible evaluation criteria.  Our simulations suggest that the (often neglected) latter choice can be influential.
We suggest that researchers seeking a formal test that guarantees type I error rate control while having scope for creative applications consider using RobustPALMRT.

\section*{Acknowledgments}

The authors would like to thank Dr.\ Leying Guan for the insightful discussions and feedback, and for providing code to reproduce her analyses.

\bibliographystyle{agsm}

\bibliography{bib}

\newpage

\renewcommand{\thesection}{A}
\section{Additional details for Huber regression}

\renewcommand{\thealgorithm}{A1}
\begin{algorithm}[!ht]
\caption{Huber regression with MAD scale estimation \citep{venables2013modern}}\label{alg:huber-mad-scale}
\begin{algorithmic}
\Require Rectangular data set $[Y, C]$.
 $Y \in \R^n$ are responses; $C \in \R^{n \times r}$ are covariates.
\State $R^0 \gets$ Residuals from OLS regression of $Y$ on $C$.
\While{Not converged}
    \State $s^k \gets 1.4826 * \text{median}_{i = 1, \ldots, n}(\text{abs}(R_i^k) )$
    \For{$i = 1, \ldots, n$}
        \State $w_i^k \gets \min \left(1, \frac{1.345}{\text{abs}(R_i^k)/s^k}\right)$
    \EndFor
    \State $R^{k + 1} \gets$ Residuals from WLS regression of $R^k$ on $C$ with weights $w^k$.
    \State Terminate if $\frac{||R^{k + 1} - R^k||_2}{||R^k||_2} < \text{Threshold}$
\EndWhile
\State \Return $R^{\text{Final}}$, $s^{\text{Final}}$
\end{algorithmic}
\end{algorithm}

\renewcommand{\thealgorithm}{A2}
\begin{algorithm}[!ht]
\caption{Huber regression with fixed scale \citep{venables2013modern}}\label{alg:huber-no-scale}
\begin{algorithmic}
\Require
$Y \in \R^n$ are responses; $C \in \R^{n \times r}$ are covariates; $s \in \R$, $s > 0$ is a fixed scale.
\State $R^0 \gets$ Residuals from OLS regression of $Y$ on $C$.
\While{Not converged}
    \For{$i = 1, \ldots, n$}
        \State $w_i^k \gets \min \left(1, \frac{1.345}{\text{abs}(R_i^k)/s}\right)$
    \EndFor
    \State $R^{k + 1} \gets$ Residuals from WLS regression of $R^k$ on $C$ with weights $w^k$.
    \State Terminate if $\frac{||R^{k + 1} - R^k||_2}{||R^k||_2} < \text{Threshold}$
\EndWhile
\State \Return $R^{\text{Final}}$, $s^{\text{Final}}$
\end{algorithmic}
\end{algorithm}

Our suggested Huber regression based approach relies on the ability to perform scale estimation in a way that satisfies conditions \ref{cond:Z-shift} and \ref{cond:symm} of section 3 of the main paper. In algorithm \ref{alg:huber-mad-scale}, we provide an example of such an algorithm, which we note is the default method used for a Huber regression in the MASS package for R \citep{venables2013modern}. We prove two properties for the output of algorithm \ref{alg:huber-mad-scale}; analogous properties hold for algorithm \ref{alg:huber-no-scale} with essentially the same proof.

Let $H(C)$ be the "hat matrix" (or projection matrix) that projects a vector $Y$ onto the column space of $C$.
Let $\Pi_\pi$ denote the permutation matrix for a permutation $\pi \in S_n$.

\begin{lemma}\label{lem:huber-properties}
    Let $R_{\text{Huber}}(Y; C), s_{\text{Huber}}(Y; C)$ be the output of algorithm \ref{alg:huber-mad-scale} with covariates $C$. Then for any $\gamma \in \R^r$, we have \begin{align*}
        R_{\text{Huber}}(Y + C\gamma; C) &= R_{\text{Huber}}(Y; C), \\
        s_{\text{Huber}}(Y + C\gamma; C) &= s_{\text{Huber}}(Y; C),
    \end{align*} and for any $\pi \in S_n$, we have 
    \begin{align*}
        R_{\text{Huber}}(Y_\pi; C) &= R_{\text{Huber}}(Y; C_{\pi^{-1}})_\pi, \\
        s_{\text{Huber}}(Y_\pi; C) &= s_{\text{Huber}}(Y; C_{\pi^{-1}}).
    \end{align*} 
\end{lemma}
\begin{proof}
    For the shift invariance, notice that 
    \begin{align*}
        R^0(Y + C\gamma; C) = Y + C\gamma - H(C) (Y + C\gamma) = Y - H(C) Y + C\gamma - C\gamma = R^0(Y; C).
    \end{align*}
    Since all remaining steps of the algorithm depend on $Y$ (or $Y + C\gamma$) only through $R^0$, the result follows.

    Let $\pi \in S_n$. Note that 
    \begin{align*}
        R^0(Y_\pi; C) = Y_\pi - H(C) Y_\pi = \Pi_\pi (Y - H(C_{\pi^{-1}}) Y) = R^0(Y; C_{\pi^{-1}})_\pi.
    \end{align*}
    Suppose that $R^k(Y_\pi; C) = R^k(Y; C_{\pi^{-1}})_\pi$. Then, since $s^k$ is a symmetric function of $R^k$, we have that $s^k(Y_\pi; C) = s^k(Y; C_{\pi^{-1}})$. Furthermore, the vector $w^k$ is symmetric in permutations of $R^k$ and $s^k$, and so $w^k(Y_\pi; C) = w^k(Y; C_{\pi^{-1}})_\pi$. Letting $H_w(C) = I - C (C^T W C)^{-1} C^T W$, we have that
    \begin{align*}
        R^{k + 1}(Y_\pi; C)
        &= R^k(Y_\pi; C) - H_{w^k(Y_\pi; C)}(C) R^k(Y_\pi; C) \\
        &= \Pi_\pi \left(R^k(Y; C_{\pi^{-1}}) - H_{w^k(Y; C_{\pi^{-1}})}(C_{\pi^{-1}}) R^k(Y; C_{\pi^{-1}})\right) \\
        &= R^{k + 1}(Y; C_{\pi^{-1}})_\pi.
    \end{align*}
    Hence by induction, since $R^0(Y_\pi; C) = R^0(Y; C_{\pi^{-1}})_\pi$, we have that $R^k(Y_\pi; C) = R^k(Y; C_{\pi^{-1}})_\pi$ for all $k \geq 0$. In particular, $R^{\text{Final}}(Y_\pi; C) = R^{\text{Final}}(Y; C_{\pi^{-1}})_\pi$ and $s^{\text{Final}}(Y_\pi; C) = s^{\text{Final}}(Y; C_{\pi^{-1}})$, as desired.
\end{proof} 

To show that the output $\pvalue$ from algorithm \ref{alg:mad-scale-full} controls type I error rate, we need to show that 
\begin{align*}
    \cM(Y, X, [Z, Z_\pi]) = \Big(\hat{s}(Y, X, [Z, Z_\pi]), \mathbf{r}(Y, X, [Z, Z_\pi])\Big)
\end{align*}
satisfies conditions \ref{cond:Z-shift} and \ref{cond:symm}. In particular,
\begin{align*}
    \hat{s}(Y + [Z, Z_\pi]\gamma, X, [Z, Z_\pi]) &= s_{\text{Huber}}(Y + [Z, Z_\pi]\gamma; [Z, Z_\pi]) 
        \\&= s_{\text{Huber}}(Y; [Z, Z_\pi]) = \hat{s}(Y, X, [Z, Z_\pi]),
    \\
    \mathbf{r}(Y + [Z, Z_\pi]\gamma, X, [Z, Z_\pi]) &= \text{Sort}(R_{\text{Huber}}(Y + [Z, Z_\pi]\gamma; [X, Z, Z_\pi])) 
        \\&= \text{Sort}(R_{\text{Huber}}(Y; [X, Z, Z_\pi])) = \mathbf{r}(Y, X, [Z, Z_\pi]),
    \\
    \hat{s}(Y_\sigma, X_\sigma, [Z, Z_\pi]_\sigma) &= s_{\text{Huber}}(Y_\sigma; [Z, Z_\pi]_\sigma)
        \\&= s_{\text{Huber}}(Y; [Z, Z_\pi]) = \hat{s}(Y, X, [Z, Z_\pi]),
    \\
    \mathbf{r}(Y_\sigma, X_\sigma, [Z, Z_\pi]_\sigma) &= \text{Sort}(R_{\text{Huber}}(Y_\sigma, [X, Z, Z_\pi]_\sigma))
        \\&= \text{Sort}(R_{\text{Huber}}(Y; [X, Z, Z_\pi])) = \mathbf{r}(Y, X, [Z, Z_\pi]).
\end{align*}
Hence we can apply theorem \ref{thm:validity}.

\renewcommand{\thesection}{B}
\section{Additional details for the design of simulations}\label{sec:appendix-sims}

We run two main experiments.
One experiment aims to find the effect that the error distribution, covariate distribution, and number of covariates has on the power of the methods; we will refer to this as the factorial experiment.
The other experiment aims to find the effect that the error distribution, covariate distribution, and number of samples has on the power of the methods; we will refer to this as the sample size experiment.

For both experiments, we first generate the covariates $[X, Z]$ and $\epsilon$, and then set
\begin{align*}
    Y = X\beta + Z\theta + \epsilon.
\end{align*}
We use $\theta = 0$ in all simulations.  Since all tested methods eliminate the effect of $Z\theta$, this has no effect on the results.
We select $\beta$ such that the F-test at nominal level $\alpha = 0.05$ for the partial correlation of $X$ with $Y$ has a particular power, ranging from $\{ 0.2, 0.4, 0.6, 0.8, 0.95 \}$.
To do this, we generate $40000$ versions of $X$, $Z$, and $\epsilon$ to get a Monte-Carlo estimate for the power of the F-test for each choice of $\beta$.
We then use the Brentq root finding approach, implemented in Scipy \citep{scipy2020} to find $\beta$ such that the Monte-Carlo evaluated F-test power matches the target.
After selecting $\beta$, for both experiments we use $1000$ trials in each setting and $B = 999$ Monte-Carlo replicates for the p-values.

For selecting the covariates, $[X, Z]$, we use random designs with varying tail weights.  The tail weights allow us to understand the effect that the distribution of leverages has on the test.
In particular, designs that use heavier tailed covariates will generally have more skewed distributions of the individual case leverages.
All settings include an intercept column in $Z$, which counts as one of the dimensions of $p$ (i.e. in the $p = 2$ setting, $Z$ has one intercept column and one additional column).
We consider the following choices for $[X, Z]$:
\begin{itemize}
    \item $[X, Z] \iid N(0, 1)$ - the Normal design.
    \item $[X, Z] \iid t_3$ - the $t_3$ design.
    \item $[X, Z] \iid \text{Cauchy}(0, 1)$ - the Cauchy design.
    \item $[X, Z]$ is all $\{0, 1\}$ valued, with each row containing only an intercept and one treatment - the Balanced ANOVA design.
\end{itemize}
In the Balanced ANOVA design, the sample size $n$ is rounded down to the nearest value such that there are the same number of observations for each treatment.
We found that the empirical distribution of the leverages has approximately the same effect on the F-test as it has on RobustPALMRT's performance, and so in the main body of the paper we only reported results in the Normal design setting.

For selecting the errors, $\epsilon$, we wanted to assess how tail weight and skewness affect the various methods.
We consider the following choices for $\epsilon$:
\begin{itemize}
    \item $\epsilon \iid N(0, 1)$ - Normal errors.
    \item $\epsilon \iid t_3$ - $t_3$ errors.
    \item $\epsilon \iid \text{Cauchy}(0, 1)$ - Cauchy errors.
    \item $\epsilon \iid \text{LogNormal}(\sigma = 1)$ - Log Normal errors. 
    \item $\epsilon \sim N(0, I_n) + (-1)^{\text{Bernoulli}(0.5)} \cdot 10^4 \cdot \text{Multinomial}\left(1; \frac{1}{n}, \ldots, \frac{1}{n}\right)$ - Multinomial errors.
\end{itemize}
The Log Normal setting has individual entries $\log(\epsilon_i) \sim N(0, 1)$.
This is a moderately skewed setting, having skewness $\gamma_1 = 6.18$.
The Multinomial errors setting represents a Normal errors setting except with one extreme outlier.

Lastly, for each setting, we additionally evaluated the type I error rate by setting $\beta = 0$ and otherwise following the same procedures.

In the factorial experiment, we vary $p \in \{ 2, 6, 16 \}$, and we test all 20 design/error distribution combinations.

In the sample size experiment, we fix $p = 6$, and we test the most interesting settings from the factorial setting, varying $n \in \{ 25, 50, 100, 200, 400 \}$.
The selected settings were
\begin{itemize}
    \item Normal design, Normal errors.
    \item Normal design, $t_3$ errors.
    \item Normal design, Cauchy errors.
    \item Normal design, Log Normal errors.
    \item Normal design, Multinomial errors.
    \item Cauchy design, Normal errors.
    \item Cauchy design, $t_3$ errors.
    \item Cauchy design, Cauchy errors.
    \item Cauchy design, Log Normal errors.
    \item Cauchy design, Multinomial errors.
    \item $t_3$ design, Log Normal errors.
    \item Balanced ANOVA design, Log Normal errors.
\end{itemize}
This collection of settings includes all of the error distributions we consider along with both normal and Cauchy errors.  The Normal designs tend to have no or only a couple high leverage points, whereas the Cauchy design creates a few very high leverage cases. 
Additionally we include two more settings with Log Normal errors in a variety of settings, since the skewed errors setting is a large part of the novelty of these methods.
However, we did not find any surprising results in this direction, and so for the sake of ease of presentation we plotted results only for the Normal and Cauchy designs.

All of the experiments were run using a cluster with 40 CPUs @ 2.40 GHz.
The total runtime across all experiments was 109 hours.
Originally we ran the experiments with a fixed scale rather than using scale estimation; the results from those experiments did not make it into the paper, since the estimated scale setting is more realistic and interesting.
Code is available to reproduce all of the results, or to efficiently produce similar results in new settings.
We select seeds in an arithmetic progression, and for reproducibility our outputs include which seed was used.
On a personal computer, it is feasible to assess the power of these methods in a new setting (for example, a different design, sample size, or error distribution) that is similar to the tested ones in under an hour of compute time.

\renewcommand{\thesection}{C}
\section{Proof of theorem 1}\label{sec:proofs}

Our proof of theorem \ref{thm:validity} owes much to the development in \citet{guan2023conformal}.
We take a slightly different approach, however, in defining a more general \emph{population} p-value and showing that it controls type I error.  We then use this to derive the guarantee for a Monte-Carlo p-value that relies on the randomly selected permutations $\pi_1, \ldots, \pi_B$.
This approach follows the ideas in \citet{ramdas2023permutation} for creating permutation tests with permutations drawn from some non-uniform distribution of permutations $F$ over the group of permutations $S_n$. 

For the remainder of this section, let us assume that $H_0: \beta = 0$ holds; that is, $Y = Z\theta + \epsilon$ for some fixed $\theta \in \R^p$ and some exchangeable distribution for $\epsilon$.
Furthermore, assume that we have a model fitting procedure $\cM$ that satisfies conditions \ref{cond:Z-shift} and \ref{cond:symm}.
Our proof applies to a more general setting, where the model evaluation procedure is allowed to simultaneously consider information from both model fits that it is comparing.
We formalize this by defining a model comparison function $\Tilde{\omega}(M_1, M_2) \in [0, 1]$ that satisfies $\Tilde{\omega}(M_1, M_2) = 1 - \Tilde{\omega}(M_2, M_1)$.
The framework discussed in the paper uses 
\begin{align*}
    \Tilde{\omega}(M_1, M_2) = I(\omega(M_1) > \omega(M_2)) + \frac{1}{2} I(\omega(M_1) = \omega(M_2)).
\end{align*}
Notice that this choice handles ties in a slightly different fashion than is done in the main paper.  For ease of presentation, in the main paper we handled ties conservatively and thus the p-values in the main paper are greater than the p-values discussed here, so the theory will carry over.

Recall that for each $\pi \in S_n$, we compute
\begin{align*}
    M^{\pi}_{\text{Orig}}
        &= \cM(Y, X, [Z, Z_{\pi}]) \\
        &\overset{H_0}{=} \cM(Z\theta + \epsilon, X, [Z, Z_{\pi}]) \\
        &= \cM(\epsilon, X, [Z, Z_\pi]), \\
    M^{\pi}_{\text{Perm}}
        &= \cM(Y, X_\pi, [Z, Z_{\pi}]) \\
        &\overset{H_0}{=} \cM(Z\theta + \epsilon, X_\pi, [Z, Z_{\pi}]) \\
        &= \cM(\epsilon, X_\pi, [Z, Z_\pi]).
\end{align*}
Hence under the null hypothesis, the expression $\omega(M^{\pi}_{\text{Orig}}, M^{\pi}_{\text{Perm}})$ does not depend on the unknown effect of the control covariates, $\theta$. For any pair of permutations $\pi, \sigma \in S_n$, the exchangeability of $\epsilon$ implies that
\begin{align*}
    \Tilde{\omega}(M^{\pi}_{\text{Orig}}, M^{\pi}_{\text{Perm}}) 
        &\overset{H_0}{=} \Tilde{\omega}(\cM(\epsilon, X, [Z, Z_\pi]), \cM(\epsilon, X_\pi, [Z, Z_\pi])) \\
        &\overset{d}{=} \Tilde{\omega}(\cM(\epsilon_{\sigma^{-1}}, X, [Z, Z_\pi]), \cM(\epsilon_{\sigma^{-1}}, X_\pi, [Z, Z_\pi])) \\
        &= \Tilde{\omega}(\cM(\epsilon, X_\sigma, [Z_\sigma, Z_{\sigma \circ \pi}]), \cM(\epsilon, X_{\sigma \circ \pi}, [Z_\sigma, Z_{\sigma \circ \pi}])) \\
        &= \Tilde{\omega}(\cM(\epsilon, X_\sigma, [Z_\sigma, Z_\tau]), \cM(\epsilon, X_\tau, [Z_\sigma, Z_\tau])),
\end{align*}
where we reparameterized to $\tau = \sigma \circ \pi$. This suggests that we define the array $A \in \R^{n! \times n!}$ by
\begin{align}\label{math:A}
    A(\tau, \sigma; \epsilon) = \Tilde{\omega}(\cM(\epsilon, X_\sigma, [Z_\sigma, Z_\tau]), \cM(\epsilon, X_\tau, [Z_\sigma, Z_\tau])).
\end{align}
Notice that the p-value \eqref{p-val} satisfies the following:
\begin{align*}
    \pvalue
        &= \frac{1 + \sum_{b = 1}^B \Tilde{\omega}(M^{{\pi_b}}_{\text{Orig}}, M^{{\pi_b}}_{\text{Perm}})}{1 + B} \\
        &\overset{H_0}{=} \frac{1 + \sum_{b = 1}^B \Tilde{\omega}(\cM(\epsilon, X, [Z, Z_{\pi_b}]), \cM(\epsilon, X_{\pi_b}, [Z, Z_{\pi_b}]))}{1 + B} \\
        &= \frac{1 + \sum_{b = 1}^B A(\pi_b, Id; \epsilon)}{1 + B} \\
        &\overset{d}{=} \frac{1 + \sum_{b = 1}^B A(\pi_b, Id; \epsilon_{\sigma^{-1}})}{1 + B} \\
        &= \frac{1 + \sum_{b = 1}^B A(\sigma \circ \pi_b, \sigma; \epsilon)}{1 + B} \\
        &= \frac{1}{1 + B} + \sum_{b = 1}^B \frac{1}{1 + B} A(\sigma \circ \pi_b, \sigma; \epsilon) \\
        &\overset{B \to \infty}{\longrightarrow} \sum_{\pi \in S_n} \bP_{\pi_b} (\pi_b = \pi) A(\sigma \circ \pi, \sigma; \epsilon) \\
        &= \E_{\pi_b}[A(\sigma \circ \pi_b, \sigma; \epsilon) | \epsilon].
\end{align*}
We have not yet appealed to any properties of the randomly selected $\pi_b$. Let $\pi_b \sim F$, for some distribution $F$ on the group of permutations $S_n$, with probability mass function $f(\pi)$.
If $F$ is a uniform distribution on some subgroup $G$ of $S_n$, then $\sigma \circ \pi_b \overset{d}{=} \pi_b$, and so
\begin{align*}
    \E_{\pi_b}\big[A(\pi_b, Id; \epsilon) ~|~ \epsilon\big]
    &\overset{d}{=} \E_{\pi_b}\big[A(\pi_b, Id; \epsilon_{\sigma^{-1}}) ~|~ \epsilon\big] \\
    &=\E_{\pi_b}\big[A(\sigma \circ \pi_b, \sigma; \epsilon) ~|~ \epsilon\big] \\
    &= \E_{\pi_b}\big[A(\pi_b, \sigma; \epsilon) ~|~ \epsilon\big].
\end{align*}

Since $\E_{\pi_b}\big[A(\pi_b, \tau; \epsilon) | \epsilon\big]$ is the average of column $\tau$ of the matrix $A$, and $A(\pi_b, \tau; \epsilon) + A(\tau, \pi_b; \epsilon) = 1$, we are prompted to bound the probability that the average of a \emph{randomly} selected column is particularly small.

We now have sufficient motivation to formally state our result.
Let $F$ be a (possibly non-uniform) distribution on the group of permutations $S_n$, with probability mass function $f(\pi)$. 
Select $\sigma \overset{ind}{\sim} F$, and define the \emph{population} p-value as:
\begin{align}\label{math:random-pop-pval}
    \pvalue^{\textup{Pop}}(\sigma)
        &= \E_{\pi_b \sim F}\big[A(\sigma^{-1} \circ \pi_b, Id; \epsilon) ~|~ \epsilon \big] \\
        &\overset{H_0}{=} \E_{\pi_b \sim F}\big[A(\sigma^{-1} \circ \pi_b, Id; Y) ~|~ \epsilon \big]
\end{align}
Furthermore, if we select $\sigma, \pi_1, \ldots, \pi_B \iid F$, then we define the \emph{Monte-Carlo} p-value as:
\begin{align}\label{math:random-MC-pval}
    \pvalue^{\textup{MC}}(\sigma)
        &= \frac{1 + \sum_{b = 1}^B A(\sigma^{-1} \circ \pi_b, Id; \epsilon)}{1 + B} \\
        &\overset{H_0}{=} \frac{1 + \sum_{b = 1}^B A(\sigma^{-1} \circ \pi_b, Id; Y)}{1 + B}].
\end{align}
Then we have the following result:
\begin{theorem}\label{thm:pop-validity}
    Suppose $\cM$ satisfies conditions \ref{cond:Z-shift}, \ref{cond:symm} and $\Tilde{\omega}$ satisfies $\Tilde{\omega}(M_1, M_2) + \Tilde{\omega}(M_2, M_2) = 1$. 
    Let $Y$ be generated according to model \eqref{regression-setup},
    and let $\sigma, \pi_1, \ldots, \pi_B \iid F$ for any (possibly non-uniform) distribution $F$ over the group of permutations $S_n$.
    Then for any $\theta \in \R^p$ and any exchangeable $\epsilon$, under the null hypothesis $H_0$ we have:
    \begin{align*}
        &\bP_{H_0}
        \left(\pvalue^{\textup{Pop}}(\sigma) \leq \alpha \right) \leq 2\alpha, \\
        &\bP_{H_0}
        \left(\pvalue^{\textup{MC}}(\sigma) \leq \alpha \right) \leq 2\alpha,
    \end{align*}
    where the randomness is averaged over $\epsilon$, $\sigma$, and $\pi_1, \ldots, \pi_B$.

    If, additionally, $F$ is uniform distribution over some subgroup $G \subseteq S_n$, then
    \begin{align*}
        \pvalue^{\textup{Pop}}(\sigma) &\overset{d}{=} \pvalue^{\textup{Pop}}(Id), \\
        \pvalue^{\textup{MC}}(\sigma) &\overset{d}{=} \pvalue^{\textup{MC}}(Id),
    \end{align*}
    and so we have 
    \begin{align*}
        &\bP_{H_0}
        \left(\pvalue^{\textup{Pop}}(Id) \leq \alpha \right) \leq 2\alpha, \\
        &\bP_{H_0}
        \left(\pvalue^{\textup{MC}}(Id) \leq \alpha \right) \leq 2\alpha,
    \end{align*}
    where the randomness includes only $\epsilon$ and $\pi_1, \ldots, \pi_B$.
\end{theorem}

Theorem \ref{thm:validity} is a special case of this result, because the p-value \eqref{p-val} in the main text  is $\pvalue^{\textup{MC}}(Id)$ and $F$ is the uniform distribution on $S_n$.
In the case of a uniform distribution on a subgroup, we do not need to generate the random permutation $\sigma \sim F$, and so the only noise in our p-value is the Monte-Carlo noise due to $\{ \pi_1, \ldots, \pi_B \}$; this is relatively benign, since the p-value will converge to the non-randomized \emph{population} p-value $\pvalue^{\textup{Pop}}(Id)$ as $B$ tends to infinity.
If $F$ is non-uniform, then the test is genuinely a randomized test, which may be undesirable.
We note that to reduce the variance due to this randomization, one can sample many choices for $\sigma \sim F$ and average the corresponding p-values at the cost of an additional factor of $2$ in the type I error control guarantee \citep{vovk2020combining}.
More precisely, we can sample $\sigma_1, \ldots, \sigma_M \iid F$ and get
\begin{align*}
    \bP_{H_0} \left(\frac{1}{M} \sum_{j = 1}^M \pvalue^{\textup{Pop}}(\sigma_j) \leq \alpha \right) \leq 4\alpha.
\end{align*}
See \citet{ramdas2023permutation} for more discussion on this.
It is unclear whether the benefit of selecting a non-uniform distribution for $F$ outweighs the cost of a substantially worse bound on the type I error.

\begin{proof}
For notational clarity, we write subscripts on $\bP$ to denote which variables are random in the given expressions.
We first control the population type I error rate:
\begin{align*}
    \bP_{\epsilon, \sigma \sim F} \Big(\pvalue^{\textup{Pop}}(\sigma) \leq \alpha \Big)
        &= \bP_{\epsilon, \sigma \sim F}
            \Big(\E_{\pi_b \sim F} \big[A(\sigma^{-1} \circ \pi_b, Id; \epsilon) ~\big|~ \epsilon \big] \leq \alpha  \Big) \\
        &= \sum_{\tau \in S_n} \bP_{\epsilon, \sigma \sim F}
            \Bigg(\sigma = \tau, \E_{\pi_b \sim F} \big[A(\tau^{-1} \circ \pi_b, Id; \epsilon) ~\big|~ \epsilon \big] \leq \alpha\Bigg) \\
        \text{(Indep. of $\sigma$ and $\epsilon$)} &= \sum_{\tau \in S_n}
            \bP_{\sigma \sim F} \big(\sigma = \tau\big)
            \bP_{\epsilon} \bigg( \E_{\pi_b \sim F} \big[A(\tau^{-1} \circ \pi_b, Id; \epsilon) ~\big|~ \epsilon \big] \leq \alpha \bigg) \\
        \text{(Symmetry of $\cM$)} &= \sum_{\tau \in S_n}
            \bP_{\sigma \sim F} \big(\sigma = \tau\big)
            \bP_{\epsilon} \bigg( \E_{\pi_b \sim F} \big[A(\pi_b, \tau; \epsilon_\tau) ~\big|~ \epsilon \big] \leq \alpha \bigg) \\
        \text{(Exch. of $\epsilon$)}&= \sum_{\tau \in S_n}
            \bP_{\sigma \sim F} \big(\sigma = \tau\big)
            \bP_{\epsilon} \bigg( \E_{\pi_b \sim F} \big[A(\pi_b, \tau; \epsilon) ~\big|~ \epsilon \big] \leq \alpha \bigg) \\
        &= \sum_{\tau \in S_n} f(\tau) \E_{\epsilon} \left[
            I\left( \sum_{\pi \in S_n} f(\pi) A( \pi, \tau; \epsilon) \leq \alpha \right)
        \right] \\
        &= \E_{\epsilon} \left[ 
            \sum_{\tau \in S_n} f(\tau) I\left( \sum_{\pi \in S_n} f(\pi) A(\pi, \tau; \epsilon) \leq \alpha \right)
        \right] \\
        \text{(Lemma \ref{lem:weight})} &\leq \E_{\epsilon} [2\alpha] \\
        &= 2\alpha.
\end{align*}
As an aside, note that the upcoming lemma \ref{lem:weight} provides a deterministic bound on the amount of weight that can be placed on columns with particularly small weighted column sum, in any array $A_{ij}$ with entries in $[0, 1]$ that satisfy $A_{ij} + A_{ji} = 1$.
While it is tempting to try to directly appeal to column exchangeability properties, only the values in the column $A(\pi, Id; \epsilon) = A(\pi; Id;, Y)$ will be computable from the data.  In order to compute $A(\pi, \sigma; \epsilon)$ for some $\sigma \neq Id$, we would need to know the true value of $\epsilon$.
Nonetheless, the $A_{ij} + A_{ji} = 1$ condition lets us bound how many columns can have small column sums, and we can use the symmetry of the matrix $A(\pi, \sigma; \epsilon)$ to relate properties of a randomly selected column to properties of our particular p-value.

Continuing the proof in the style of \citet{ramdas2023permutation}, we now control the type I error for the Monte-Carlo version.
Let $\sigma, \pi_1, \ldots, \pi_B \iid F$, and for ease of notation let $\pi_0 = \sigma$.
Notice that the vector $(\pi_0, \pi_1, \ldots, \pi_B)$ is exchangeable, and thus for any $\tau \in S_{B+1}$, a permutation of the set $\{ 0, \ldots, B \}$, conditioned on the multiset $\{ \pi_0, \ldots, \pi_B \}$ we have
\begin{align*}
    (\pi_{\tau(0)}, \pi_{\tau(1)}, \ldots, \pi_{\tau(B)}) \overset{d}{=} (\pi_0, \pi_1, \ldots, \pi_B).
\end{align*}
In particular, notice that if we define, for any $\Tilde{\sigma}, \Tilde{\pi_0}, \ldots, \Tilde{\pi_B}$,
\begin{align*}
    g(\Tilde{\sigma}, \Tilde{\pi}_0, \ldots, \Tilde{\pi}_B) &= \frac{\sum_{b = 0}^B A(\Tilde{\sigma}^{-1} \circ \Tilde{\pi}_b, Id; \epsilon)}{1 + B},
\end{align*}
then we have (conditioned on $\{ \pi_0, \ldots, \pi_B \}$):
\begin{align}\label{math:g-exch}
    g(\pi_{\tau(0)}, \pi_{\tau(0)}, \ldots, \pi_{\tau(B)}) \overset{d}{=} g(\pi_0, \pi_0, \ldots, \pi_B).
\end{align}
Furthermore, for any $\tau \in S_{B + 1}$, we have
\begin{align}\label{math:g-symm}
    \nonumber
    g(\Tilde{\sigma}, \Tilde{\pi}_{\tau(0)}, \ldots, \Tilde{\pi}_{\tau(B)}) &= \frac{\sum_{b = 0}^B A(\Tilde{\sigma}^{-1} \circ \Tilde{\pi}_{\tau(b)}, Id; \epsilon)}{1 + B} \\
    \nonumber
    &= \frac{\sum_{b = 0}^B A(\Tilde{\sigma}^{-1} \circ \Tilde{\pi}_b, Id; \epsilon)}{1 + B} \\
    &= g(\Tilde{\sigma}, \Tilde{\pi}_0, \ldots, \Tilde{\pi}_B).
\end{align}
Equipped with these two facts, we will first describe the p-values in terms of $g$, and then we will appeal to the validity of the population p-value for \emph{any} choice of $F$.
In particular, define the empirical distribution $\{ \pi_0, \ldots, \pi_B \}$:
\begin{align*}
    \Tilde{F} = \frac{1}{B + 1} \sum_{b = 0}^B \delta_{\pi_b}.
\end{align*}
Then selecting $K \sim \text{Unif}\{0, \ldots, B\}$, we have $\pi_K \sim \Tilde{F}$. Then by the population validity result,
\begin{align*}
    \bP_{\epsilon, K, \pi_0, \ldots, \pi_B } \left(\pvalue^{\textup{Pop}}(\pi_K) \leq \alpha ~\big|~ \{ \pi_0, \ldots, \pi_B \} \right) \leq 2\alpha.
\end{align*}
Notice that
\begin{align*}
    \pvalue^{\textup{Pop}}(\pi_K)
        &= \E_{\Tilde{\pi} \sim \Tilde{F}} \left[
            A(\pi_K^{-1} \circ \Tilde{\pi}, Id; \epsilon) ~|~ \epsilon
        \right] \\
        &= \frac{\sum_{b = 0}^B A(\pi_K^{-1} \circ \pi_b, Id; \epsilon)}{B + 1} \\
        &= g(\pi_K, \pi_0, \ldots, \pi_B),
\end{align*}
and lastly notice that
\begin{align*}
    \pvalue^{\textup{MC}}(\sigma) &= \pvalue^{\textup{MC}}(\pi_0) \\
    &= \frac{1 + \sum_{b = 1}^B A(\pi_0^{-1} \circ \pi_b, Id; \epsilon)}{1 + B} \\
    &= \frac{\frac{1}{2} + \sum_{b = 0}^B A(\pi_0^{-1} \circ \pi_b, Id; \epsilon)}{1 + B} \\
    &= g(\pi_0, \pi_0, \ldots, \pi_B) + \frac{1}{2(1 + B)}.    
\end{align*}
Calling $\tau_k \in S_{B+1}$ the permutation that swaps $0$ and $k$ for any $k \in \{0, \ldots, B\}$, we have
\begin{align*}
    \bP_{\epsilon, K, \pi_0, \ldots, \pi_B }\Big(&\pvalue^{\textup{MC}}(\pi_0) \leq \alpha ~\big|~ \{ \pi_0, \ldots, \pi_B \} \Big) \\
    &= \bP_{\epsilon, K, \pi_0, \ldots, \pi_B }\bigg(g\big(\pi_0, \pi_0, \ldots, \pi_B\big) + \frac{1}{2(1 + B)} \leq \alpha ~\big|~ \{ \pi_0, \ldots, \pi_B \} \bigg) \\
    &\leq \bP_{\epsilon, K, \pi_0, \ldots, \pi_B }\Big(g\big(\pi_0, \pi_0, \ldots, \pi_B\big) \leq \alpha ~\big|~ \{ \pi_0, \ldots, \pi_B \} \Big) \\
    \text{(By \eqref{math:g-exch})} &= \bP_{\epsilon, K, \pi_0, \ldots, \pi_B }\Big(g\big(\pi_{\tau_K(0)}, \pi_{\tau_K(0)}, \ldots, \pi_{\tau_K(B)}\big) \leq \alpha ~\big|~ \{ \pi_0, \ldots, \pi_B \} \Big) \\
    \text{(By \eqref{math:g-symm})} &= \bP_{\epsilon, K, \pi_0, \ldots, \pi_B }\Big(g\big(\pi_{\tau_K(0)}, \pi_0, \ldots, \pi_B \big) \leq \alpha ~\big|~ \{ \pi_0, \ldots, \pi_B \} \Big) \\
    &= \bP_{\epsilon, K, \pi_0, \ldots, \pi_B }\Big(g\big(\pi_K, \pi_0, \ldots, \pi_B \big) \leq \alpha ~\big|~ \{ \pi_0, \ldots, \pi_B \} \Big) \\
    &= \bP_{\epsilon, K, \pi_0, \ldots, \pi_B } \left(\pvalue^{\textup{Pop}}(\pi_K) \leq \alpha ~\big|~ \{ \pi_0, \ldots, \pi_B \} \right) \\
    &\leq 2\alpha.
\end{align*}
Marginalizing over $\{ \pi_0, \ldots, \pi_B \}$, we have
\begin{align*}
    \bP_{\epsilon, K, \pi_0, \ldots, \pi_B }\Big(\pvalue^{\textup{MC}}(\pi_0) \leq \alpha \Big) &= \E_{\epsilon, K, \pi_0, \ldots, \pi_B }\left[\bP_{\epsilon, K, \pi_0, \ldots, \pi_B }\Big(\pvalue^{\textup{MC}}(\pi_0) \leq \alpha ~\big|~ \{ \pi_0, \ldots, \pi_B \} \Big)\right] \\
    &\leq E[2\alpha] \\
    &= 2\alpha,
\end{align*}
as desired.
\end{proof}

We now present the combinatorial bound on the weighted column sums. This is very similar to a proof in appendix E.4.3 of \citet{barber2023conformal}, but this statement is slightly more general so we reproduce it here:
\begin{lemma}{\label{lem:weight}}
    Let $A \in [0, 1]^{m \times m}$ be a matrix with the property that
    \begin{align*}
        A_{ij} + A_{ji} = 1
    \end{align*}
    for all $i, j \in [m]$. Furthermore, let $w_1, \ldots, w_m \in [0, 1]$ have $\sum_{j = 1}^m w_j = 1$. Fix $\alpha \in [0, 1/2]$ and define
    \begin{align*}
        S &= \left\{ i \in [m] ~|~ \sum_{j = 1}^m w_j A_{ij} \geq 1 - \alpha \right\} \\
        &= \left\{ i \in [m] ~|~ \sum_{j = 1}^m w_j A_{ji} \leq \alpha \right\}.
    \end{align*}
    Then
    \begin{align*}
        \sum_{i \in [m]} w_i I \left(\sum_{j \in [m]} w_j A_{ji} \leq \alpha\right) = \sum_{i \in S} w_i \leq 2\alpha.
    \end{align*}
\end{lemma}
\begin{proof}
    For any $i \in S$,
    \begin{align*}
        1 - \alpha &\leq \sum_{j = 1}^m w_j A_{ij} \\
        &\leq \sum_{j \in S} w_j A_{ij} + \sum_{j \not\in S} w_j \\
        &= \sum_{j \in S} w_j A_{ij} + 1 - \sum_{j \in S} w_j.
    \end{align*}
    Then notice that
    \begin{align*}
        \sum_{i,j \in S} w_i w_j A_{ij} = \sum_{i,j \in S} w_i w_j A_{ji},
    \end{align*}
    and so
    \begin{align*}
        \sum_{i,j \in S} w_i w_j A_{ij} &= \frac{1}{2} \sum_{i,j \in S} w_i w_j \left[A_{ij} + A_{ji}\right] \\  
        &\leq \frac{1}{2} \sum_{i,j \in S} w_i w_j.
    \end{align*}
    Hence
    \begin{align*}
        (1 - \alpha) \sum_{i \in S} w_i &\leq \sum_{i \in S} w_i \left[ \sum_{j \in S} w_j A_{ij} + 1 - \sum_{j \in S} w_j \right] \\
        &= \sum_{i,j \in S} w_i w_j A_{ij} + \sum_{i \in S} w_i - \sum_{i,j \in S} w_i w_i \\
        &\leq \frac{1}{2} \sum_{i,j \in S} w_i w_j + \sum_{i \in S} w_i - \sum_{i,j \in S} w_i w_i \\
        &= -\frac{1}{2} \sum_{i,j \in S} w_i w_j + \sum_{i \in S} w_i \\
        &= -\frac{1}{2} \left(\sum_{i \in S} w_i \right)^2  + \sum_{i \in S} w_i.
    \end{align*}
    Straightforward algebra gives $\sum_{i \in S} w_i \leq 2 \alpha$, as desired.

    To see that the two versions of $S$ are equal, notice that
    \begin{align*}
        \sum_{j \in [m]} w_j A_{ji} &= \sum_{j \in [m]} w_j (1 - A_{ij}) \\
        &= 1 - \sum_{j \in [m]} w_j A_{ij}.
    \end{align*}
\end{proof}

\begin{corollary} Let $A(\pi, \tau; \epsilon)$ be as \eqref{math:A}. For any distribution $F$ on the group of permutations $S_n$, with probability mass function $f(\pi)$, we have:
    \begin{align*}
        \sum_{\tau \in S_n} f(\tau) I \left(\sum_{\pi \in S_n} f(\pi) A(\pi, \tau; \epsilon) \leq \alpha\right) \leq 2\alpha.
    \end{align*}
\end{corollary}
\begin{proof}
    We use lemma \ref{lem:weight} with $m = n!$, $A_{\pi, \tau} = A(\pi, \tau; \epsilon)$, and $w_\pi = f(\pi)$.
\end{proof}

\renewcommand{\thesection}{D}
\section{Additional plots for simulations}\label{sec:additional-plots-sims}


\begin{figure}[!htbp]
    \centering
    \includegraphics[width=0.85\linewidth]{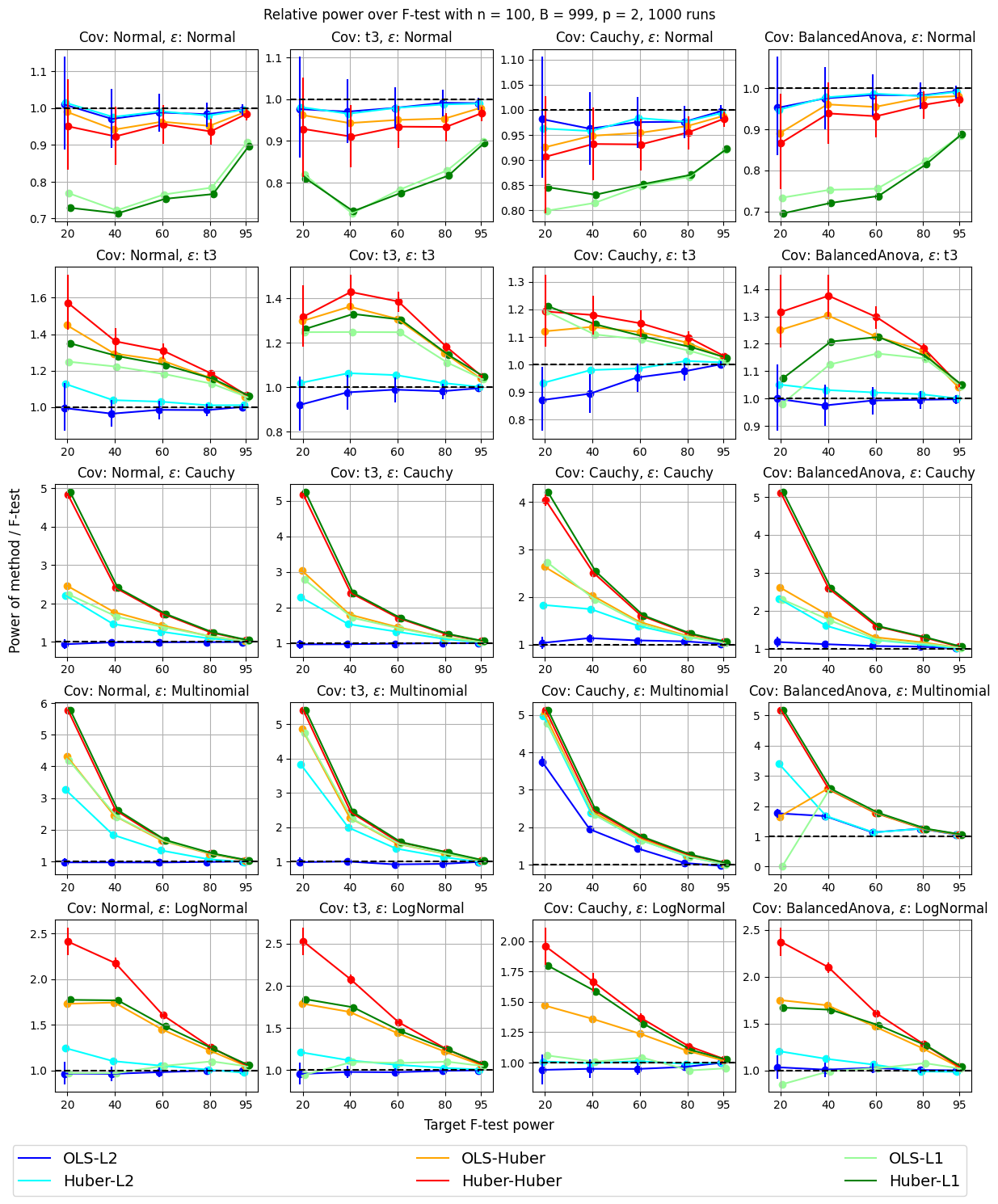}
    \vskip -0.25in
    \caption{
        Relative power for various RobustPALRMT regression approaches compared to the F-test, computed using $1000$ trials.  See section \ref{sec:Simulations}. Here we use $p = 2$ covariates in $Z$ and sample size of $n = 100$. Monte-Carlo $95\%$ error bars are plotted for the two methods of most interest, that is Huber-Huber RobustPALMRT and OLS-L2 PALMRT. The simulation is blocked, meaning for each setting and target F-test power, in a given replicate, the same data set is used for all methods. This leads to conservative error bars for the differences between the methods.
    }
    \label{fig:n100-p2-rel-power}
\end{figure}

\begin{figure}[!htbp]
    \centering
    \includegraphics[width=0.85\linewidth]{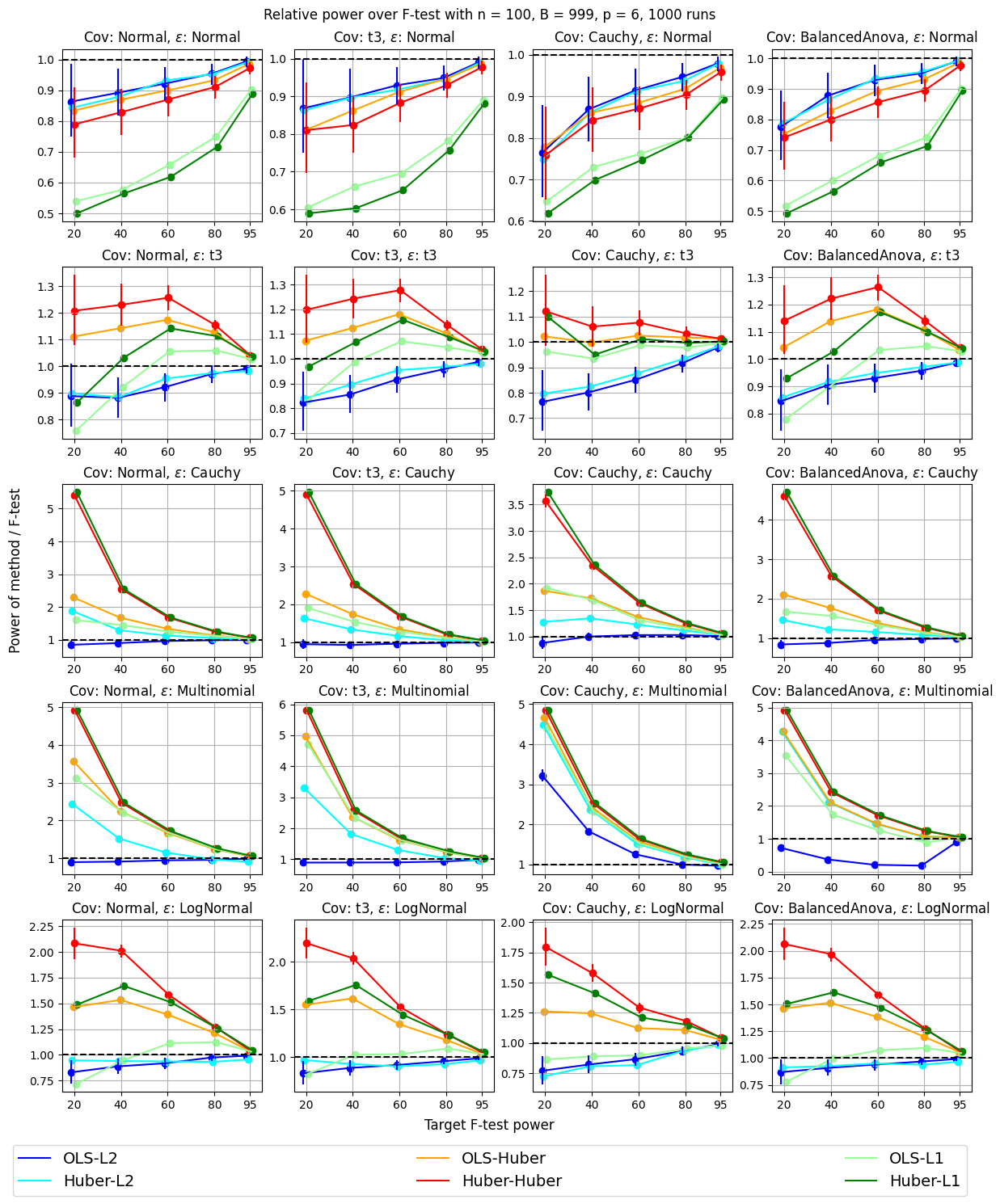}
    \vskip -0.25in
    \caption{
        Relative power for various RobustPALRMT regression approaches compared to the F-test, computed using $1000$ trials.  See section \ref{sec:Simulations}. Here we use $p = 6$ covariates in $Z$ and sample size of $n = 100$. Monte-Carlo $95\%$ error bars are plotted for the two methods of most interest, that is Huber-Huber RobustPALMRT and OLS-L2 PALMRT. The simulation is blocked, meaning for each setting and target F-test power, in a given replicate, the same data set is used for all methods. This leads to conservative error bars for the differences between the methods. 
    }
    \label{fig:n100-p6-rel-power}
\end{figure}

\begin{figure}[!htbp]
    \centering
    \includegraphics[width=0.85\linewidth]{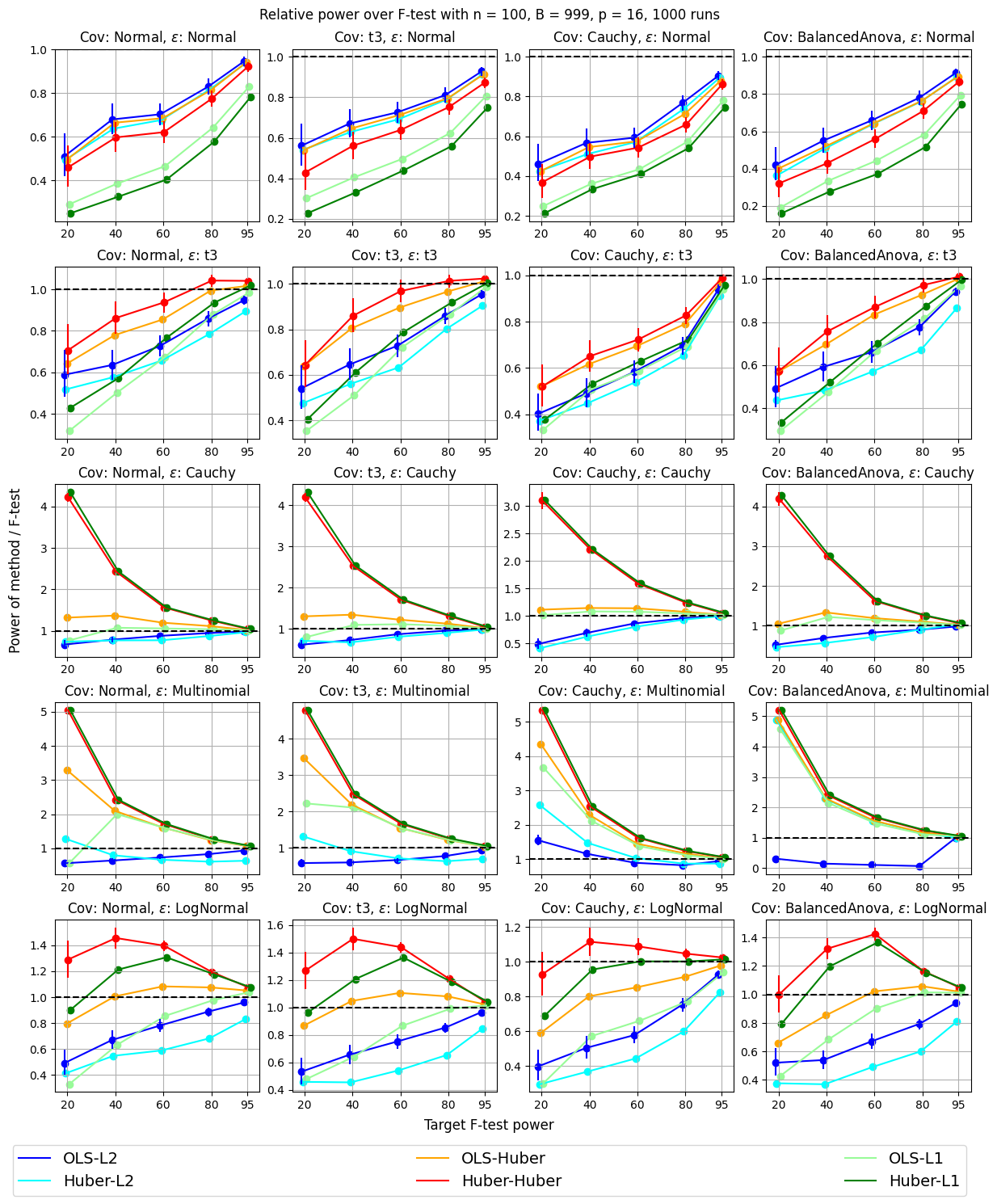}
    \vskip -0.25in
    \caption{
        Relative power for various RobustPALRMT regression approaches compared to the F-test, computed using $1000$ trials.  See section \ref{sec:Simulations}. Here we use $p = 16$ covariates in $Z$ and sample size of $n = 100$. Monte-Carlo $95\%$ error bars are plotted for the two methods of most interest, that is Huber-Huber RobustPALMRT and OLS-L2 PALMRT. The simulation is blocked, meaning for each setting and target F-test power, in a given replicate, the same data set is used for all methods. This leads to conservative error bars for the differences between the methods. 
    }
    \label{fig:n100-p16-rel-power}
\end{figure}


\begin{figure}[!htbp]
    \centering
    \includegraphics[width=0.85\linewidth]{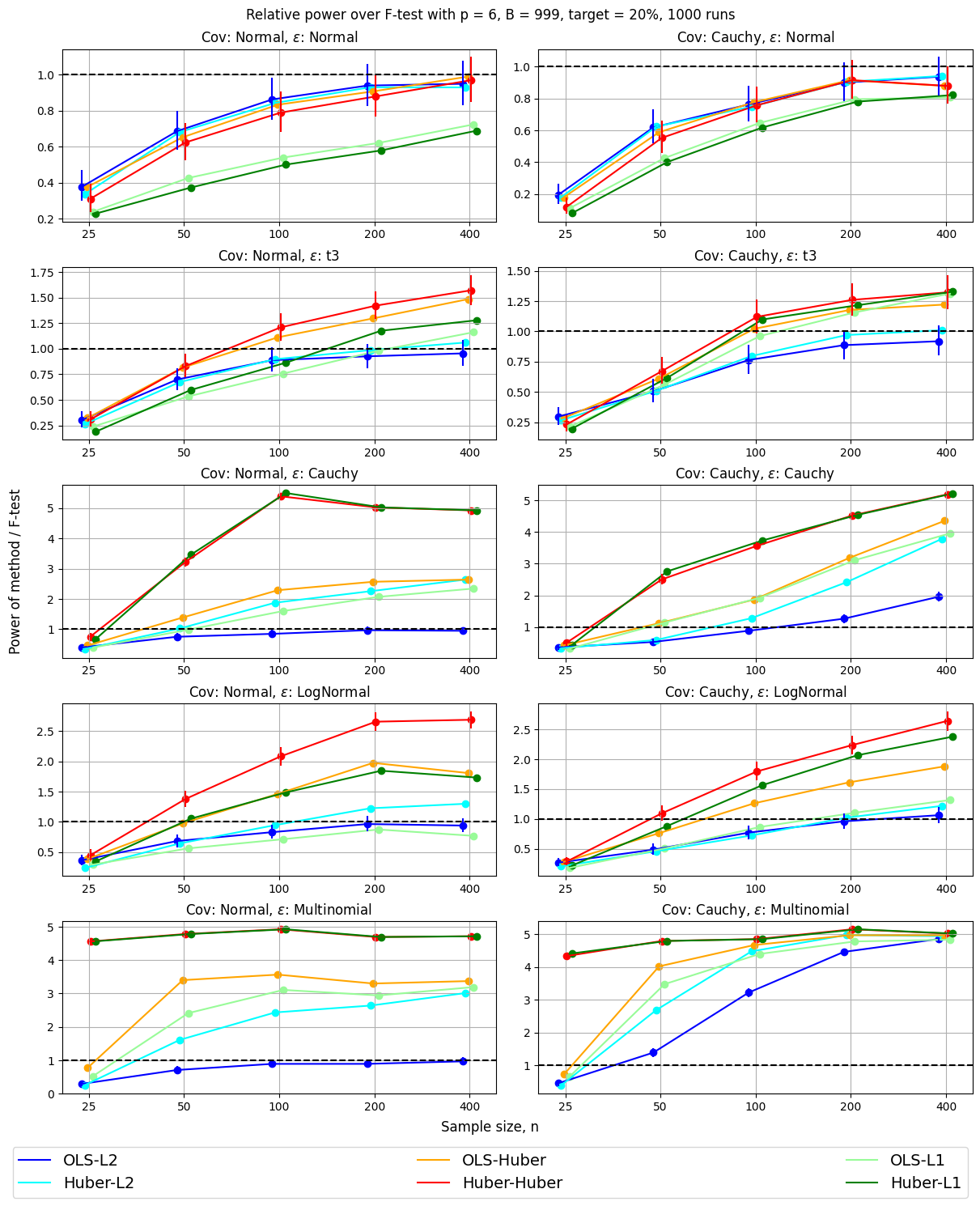}
    \vskip -0.25in
    \caption{
        Relative power of RobustPALMRT compared to the  F-test vs sample size $n$.
        In each setting, $\beta$ is selected to fix the power of the F-test at $20\%$.
        Note that the maximum possible ratio of powers is $5$.
        There are $1000$ replicates with $B = 999$.
        We show error bars for our method (Huber-Huber) and \citep{guan2023conformal} (OLS-L2), the two of most interest.
        To improve the visualization of overlapping curves and error bars, we  ``jittered'' curves horizontally, preserving the ratios and the shapes of the curves.
    }
    \label{fig:target20-rel-power}
\end{figure}

\begin{figure}[!htbp]
    \centering
    \includegraphics[width=0.88\linewidth]{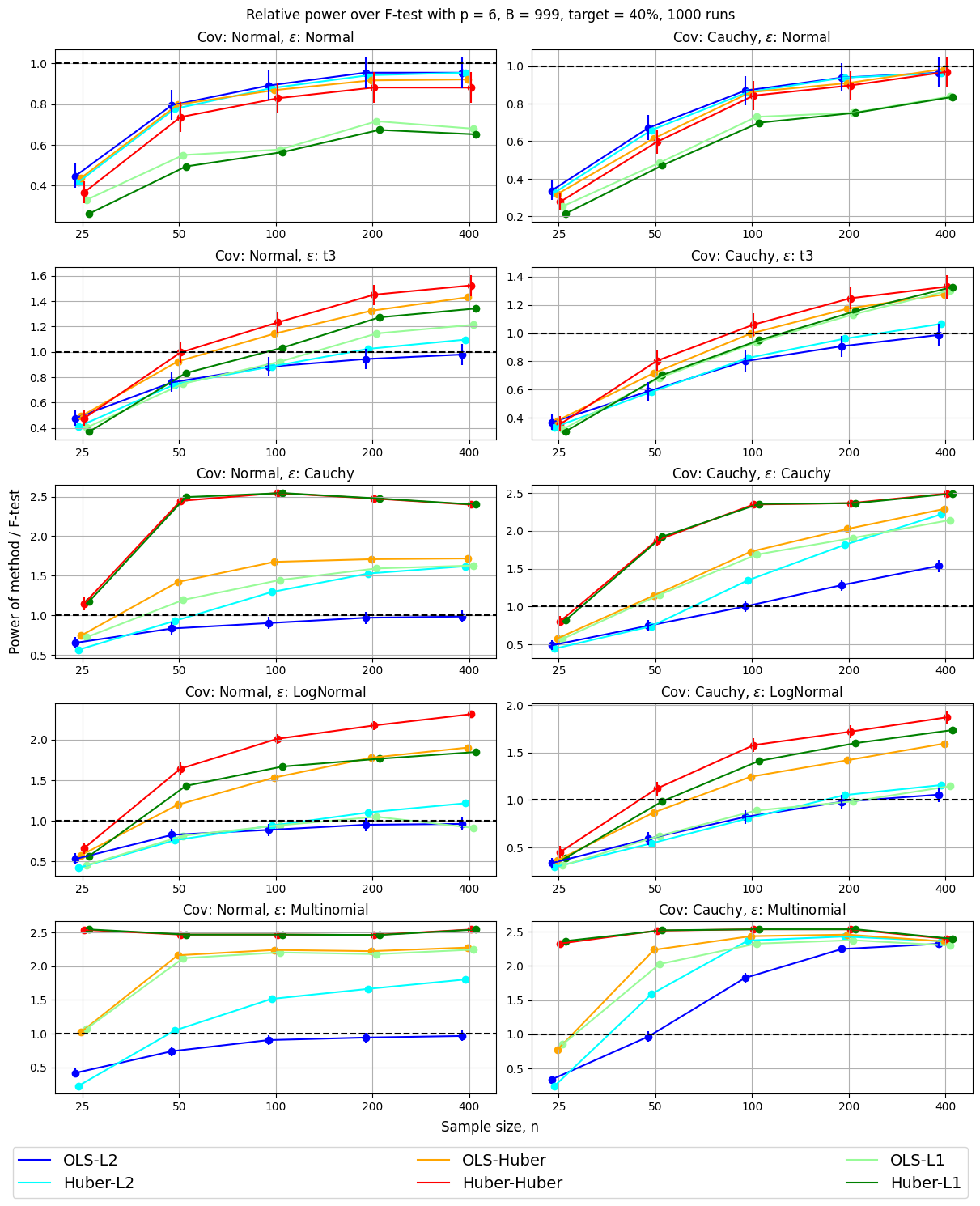}
    \vskip -0.25in
    \caption{
        Relative power of RobustPALMRT compared to the  F-test vs sample size $n$.
        In each setting, $\beta$ is selected to fix the power of the F-test at $40\%$.
        Note that the maximum possible ratio of powers is $2.5$.
        There are $1000$ replicates with $B = 999$.
        We show error bars for our method (Huber-Huber) and \citep{guan2023conformal} (OLS-L2), the two of most interest.
        To improve the visualization of overlapping curves and error bars, we  ``jittered'' curves horizontally, preserving the ratios and the shapes of the curves.
    }
    \label{fig:target40-rel-power}
\end{figure}

\begin{figure}[!htbp]
    \centering
    \includegraphics[width=0.88\linewidth]{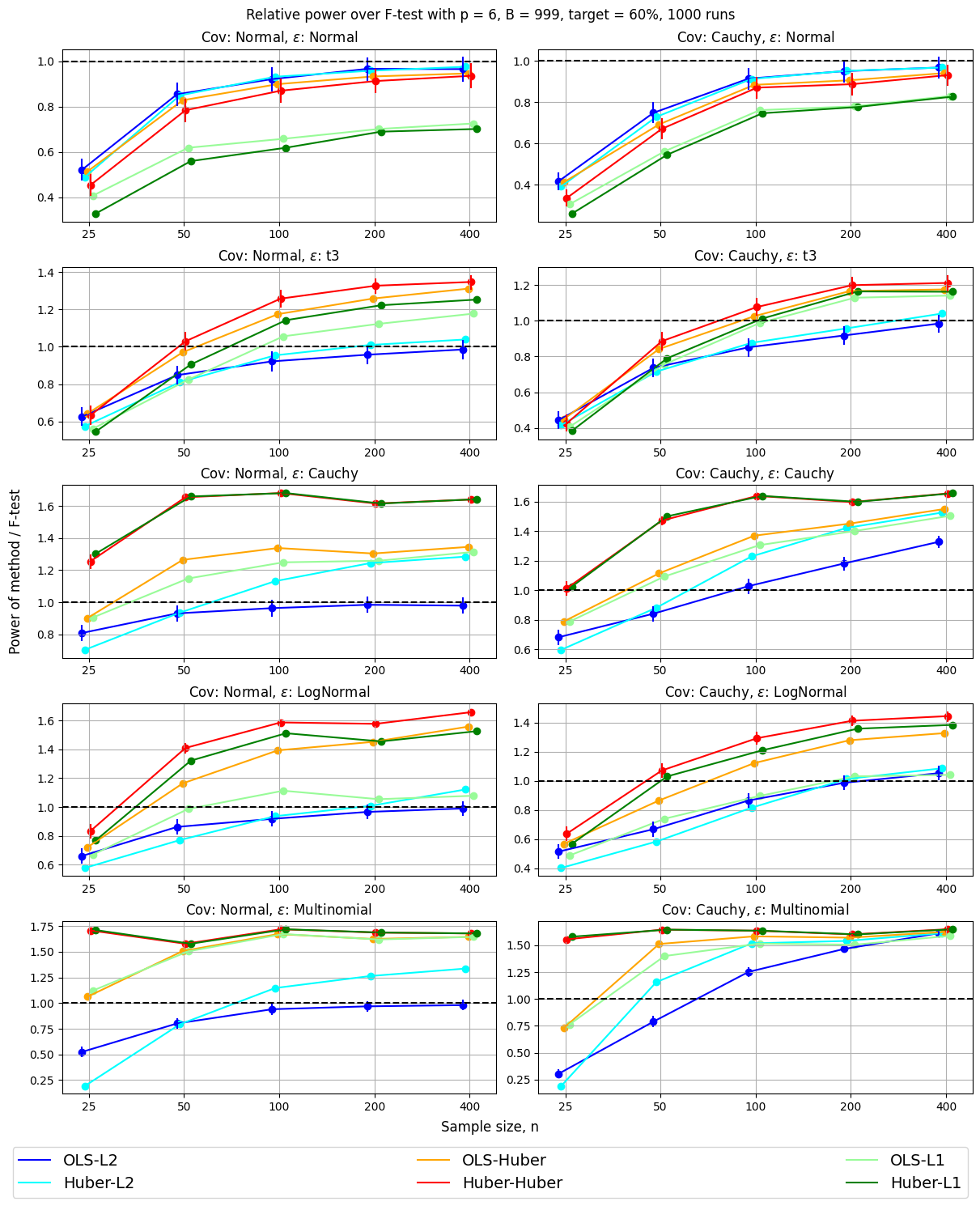}
    \vskip -0.25in
    \caption{
        Relative power of RobustPALMRT compared to the  F-test vs sample size $n$.
        In each setting, $\beta$ is selected to fix the power of the F-test at $60\%$.
        Note that the maximum possible ratio of powers is $1.67$.
        There are $1000$ replicates with $B = 999$.
        We show error bars for our method (Huber-Huber) and \citep{guan2023conformal} (OLS-L2), the two of most interest.
        To improve the visualization of overlapping curves and error bars, we  ``jittered'' curves horizontally, preserving the ratios and the shapes of the curves.
    }
    \label{fig:target60-rel-power}
\end{figure}

\begin{figure}[!htbp]
    \centering
    \includegraphics[width=0.88\linewidth]{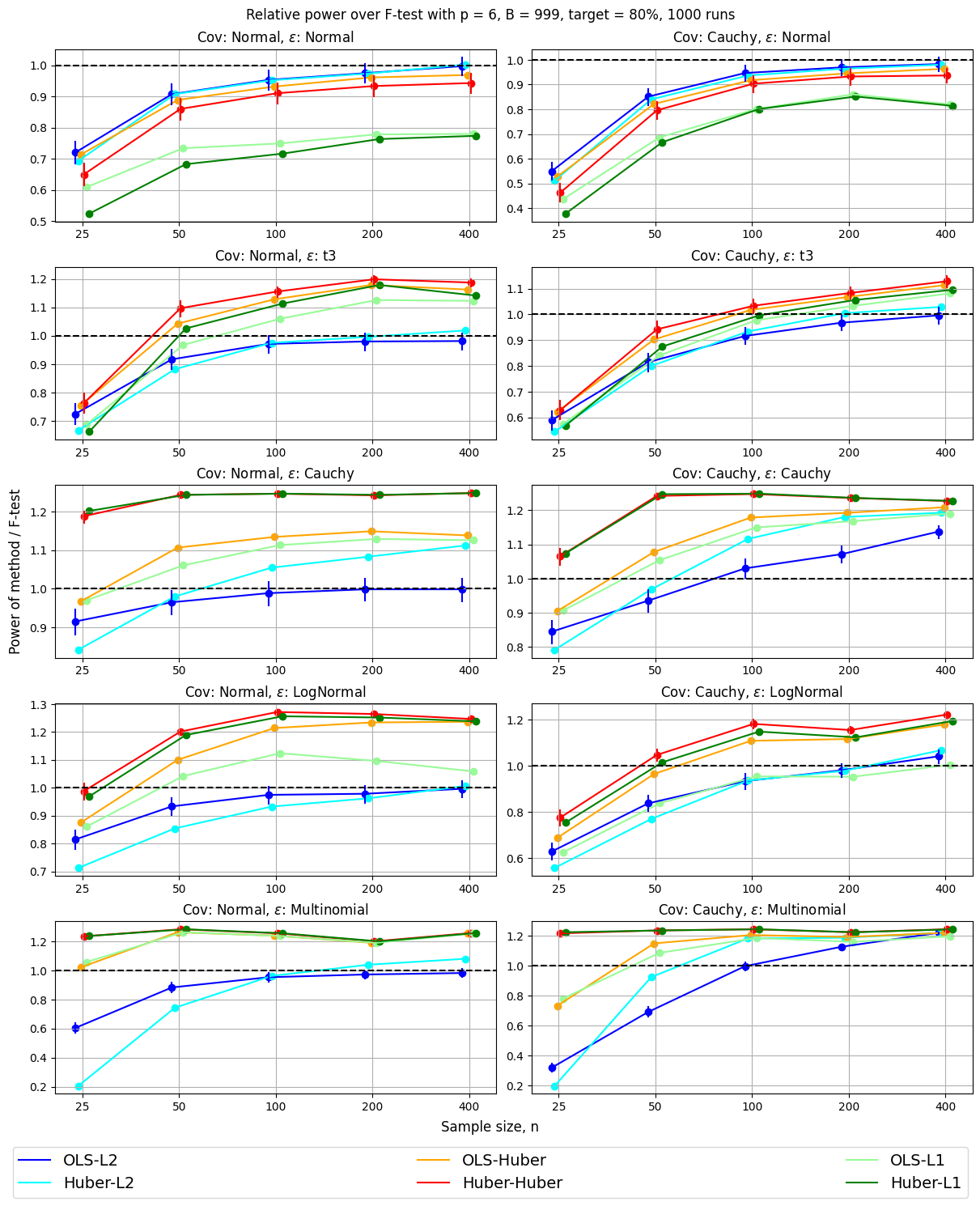}
    \vskip -0.25in
    \caption{
        Relative power of RobustPALMRT compared to the  F-test vs sample size $n$.
        In each setting, $\beta$ is selected to fix the power of the F-test at $80\%$.
        Note that the maximum possible ratio of powers is $1.25$.
        There are $1000$ replicates with $B = 999$.
        We show error bars for our method (Huber-Huber) and \citep{guan2023conformal} (OLS-L2), the two of most interest.
        To improve the visualization of overlapping curves and error bars, we  ``jittered'' curves horizontally, preserving the ratios and the shapes of the curves.
    }
    \label{fig:target80-rel-power}
\end{figure}

\begin{figure}[!htbp]
    \centering
    \includegraphics[width=0.88\linewidth]{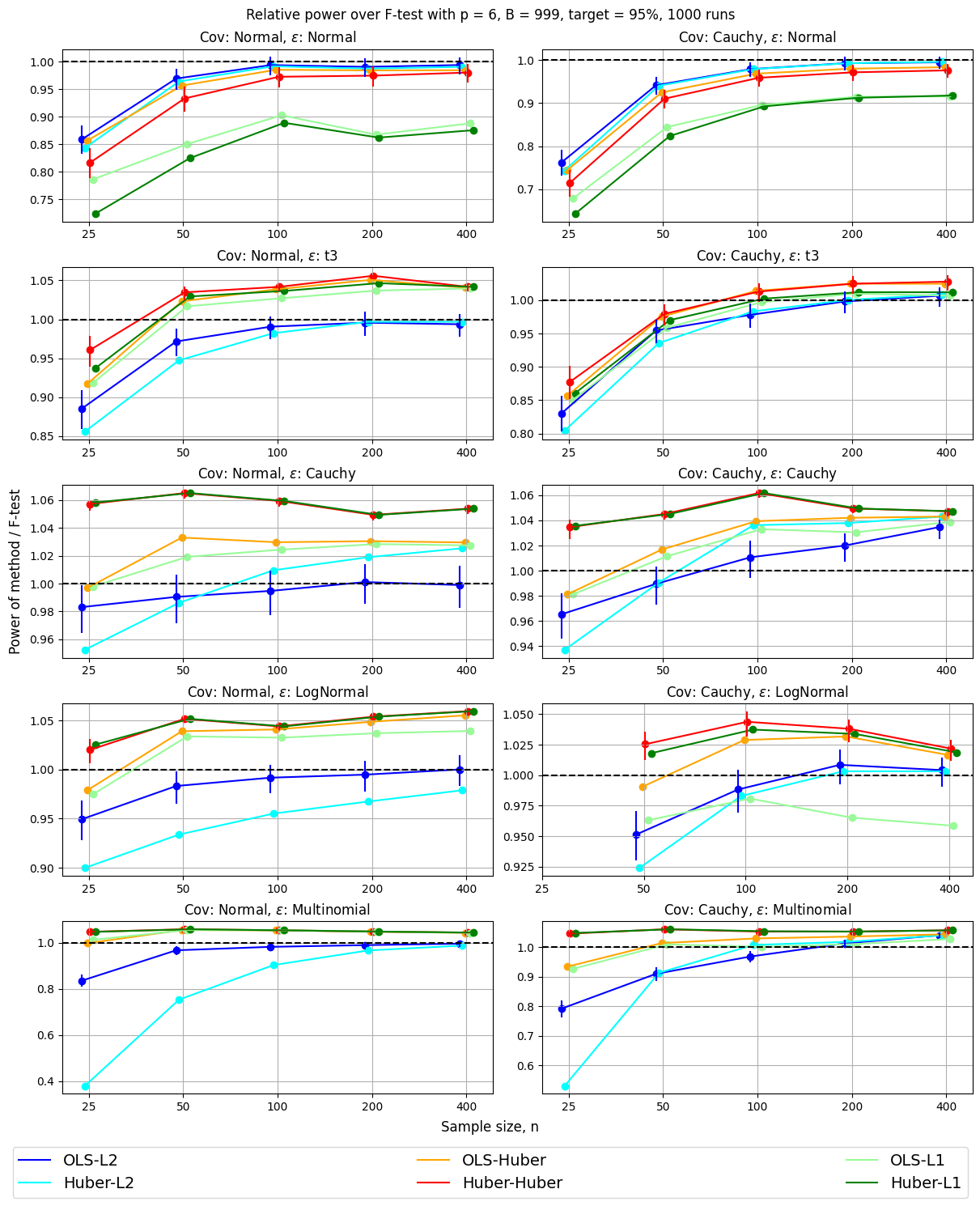}
    \vskip -0.25in
    \caption{
        Relative power of RobustPALMRT compared to the  F-test vs sample size $n$.
        In each setting, $\beta$ is selected to fix the power of the F-test at $95\%$.
        Note that the maximum possible ratio of powers is $1.05$.
        There are $1000$ replicates with $B = 999$.
        We show error bars for our method (Huber-Huber) and \citep{guan2023conformal} (OLS-L2), the two of most interest.
        To improve the visualization of overlapping curves and error bars, we  ``jittered'' curves horizontally, preserving the ratios and the shapes of the curves.
    }
    \label{fig:target95-rel-power}
\end{figure}


\begin{figure}[!htbp]
    \centering
    \includegraphics[width=0.95\linewidth]{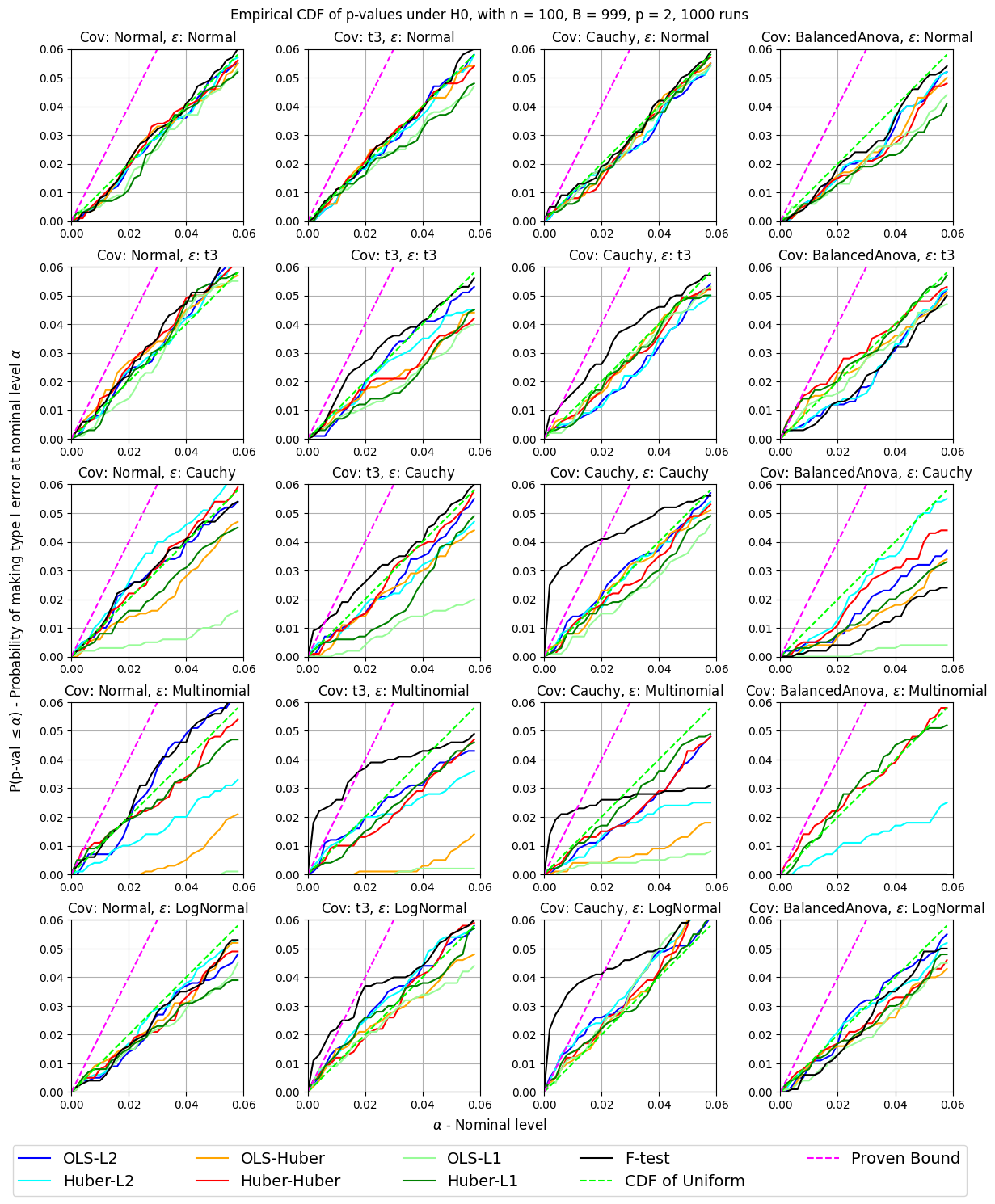}
    \vskip -0.25in
    \caption{
        Empirical CDF of RobustPALMRT and F-test p-values, with $\beta = 0$, $n = 100$, and $p = 2$.
        Ideally the empirical CDF (actual p-value) would match the Uniform CDF (nominal p-value; the green dashed line) as closely as possible.
        Our proofs ensure that the CDF of the RobustPALMRT methods lie below the $2\alpha$ line (the pink dashed line), but notice that empirically they fall at or below the Uniform CDF.
    }
    \label{fig:n100-p2-typeI-sim}
\end{figure}

\begin{figure}[!htbp]
    \centering
    \includegraphics[width=0.95\linewidth]{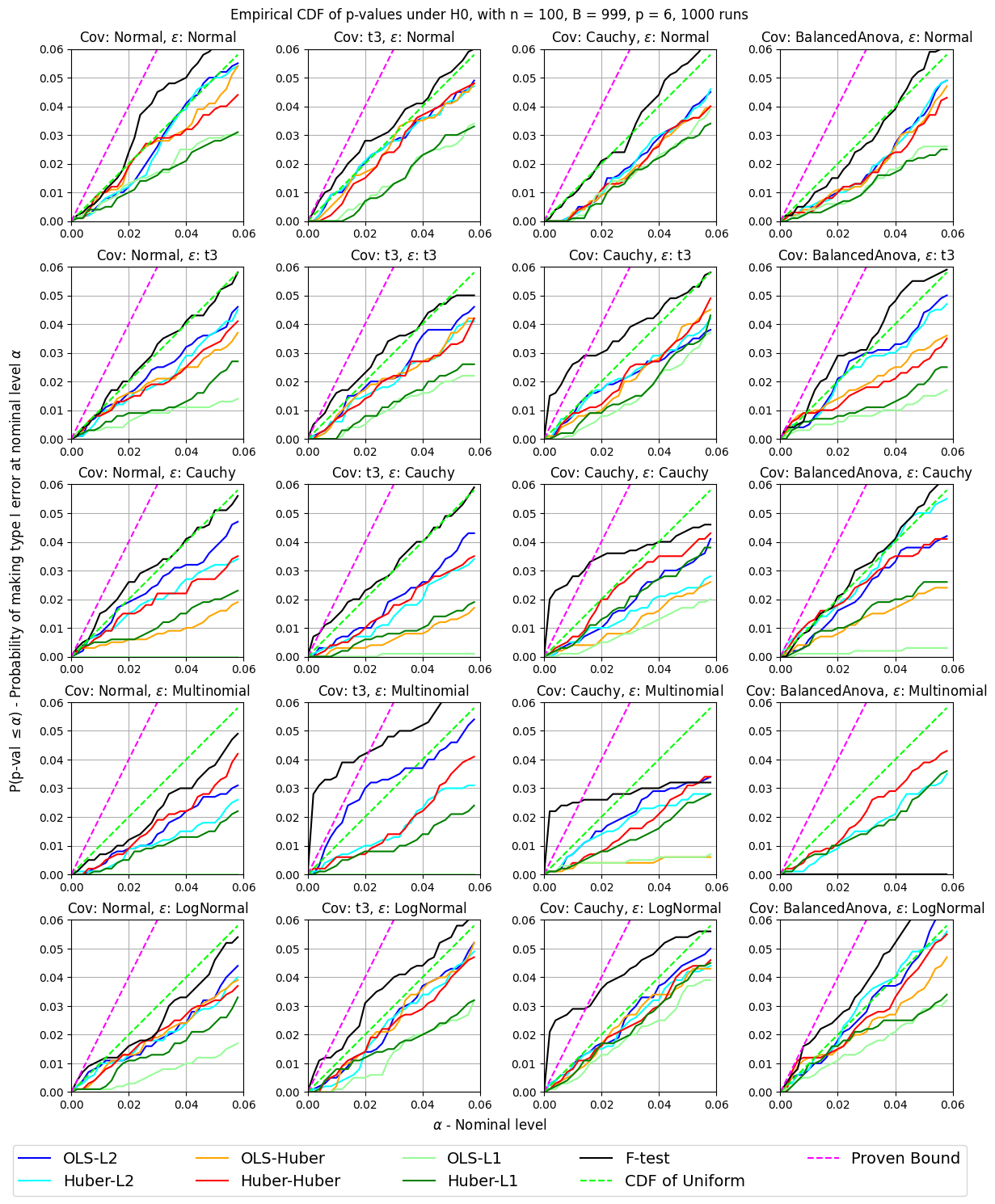}
    \vskip -0.25in
    \caption{
        Empirical CDF of RobustPALMRT and F-test p-values, with $\beta = 0$, $n = 100$, and $p = 6$.
        Ideally the empirical CDF (actual p-value) would match the Uniform CDF (nominal p-value; the green dashed line) as closely as possible.
        Our proofs ensure that the CDF of the RobustPALMRT methods lie below the $2\alpha$ line (the pink dashed line), but notice that empirically they fall at or below the Uniform CDF.
    }
    \label{fig:n100-p6-typeI-sim}
\end{figure}

\begin{figure}[!htbp]
    \centering
    \includegraphics[width=0.95\linewidth]{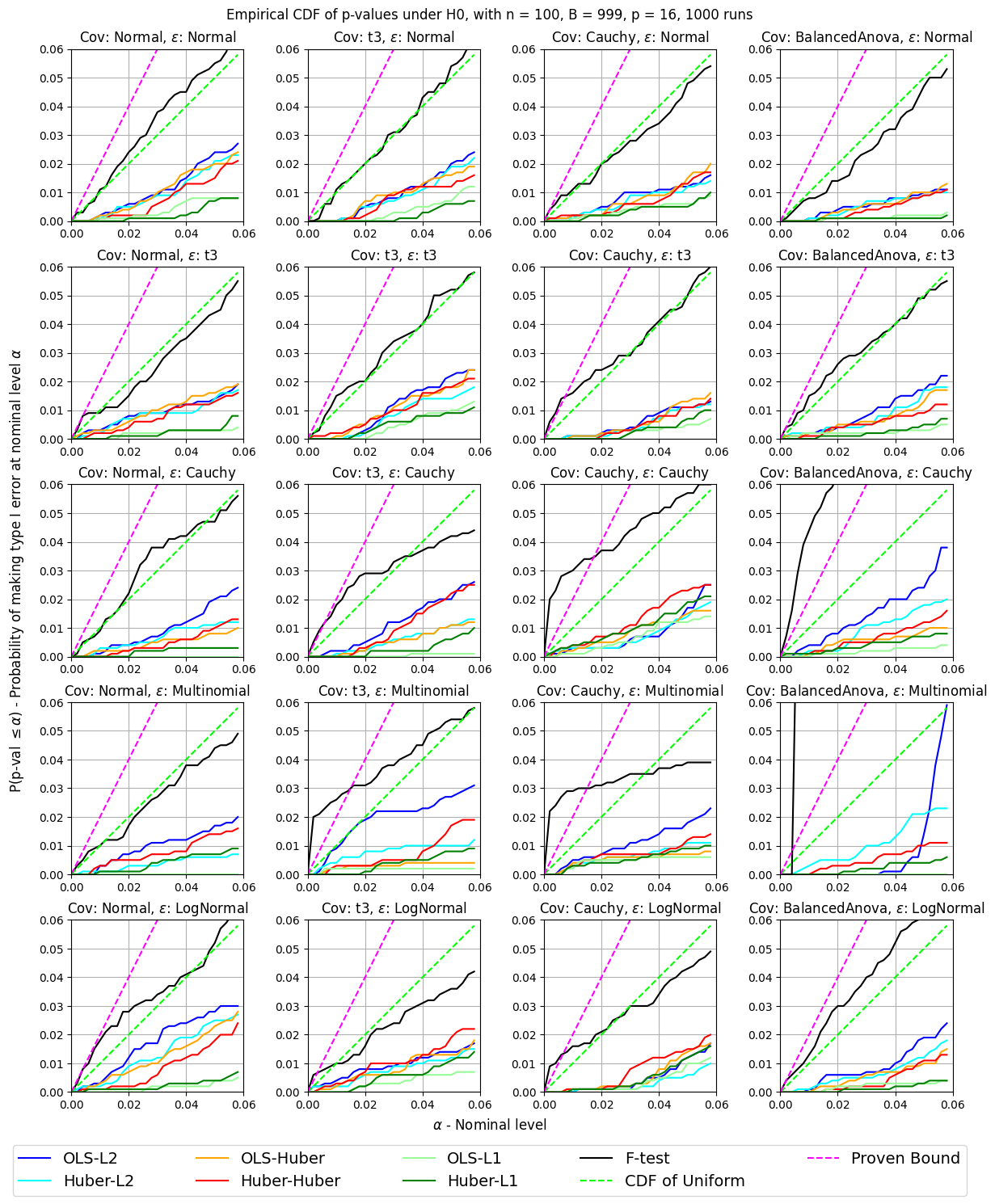}
    \vskip -0.25in
    \caption{
        Empirical CDF of RobustPALMRT and F-test p-values, with $\beta = 0$, $n = 100$, and $p = 16$.
        Ideally the empirical CDF (actual p-value) would match the Uniform CDF (nominal p-value; the green dashed line) as closely as possible.
        Our proofs ensure that the CDF of the RobustPALMRT methods lie below the $2\alpha$ line (the pink dashed line), but notice that empirically they fall at or below the Uniform CDF.
    }
    \label{fig:n100-p16-typeI-sim}
\end{figure}


\begin{figure}[!htbp]
    \centering
    \includegraphics[width=0.94\linewidth]{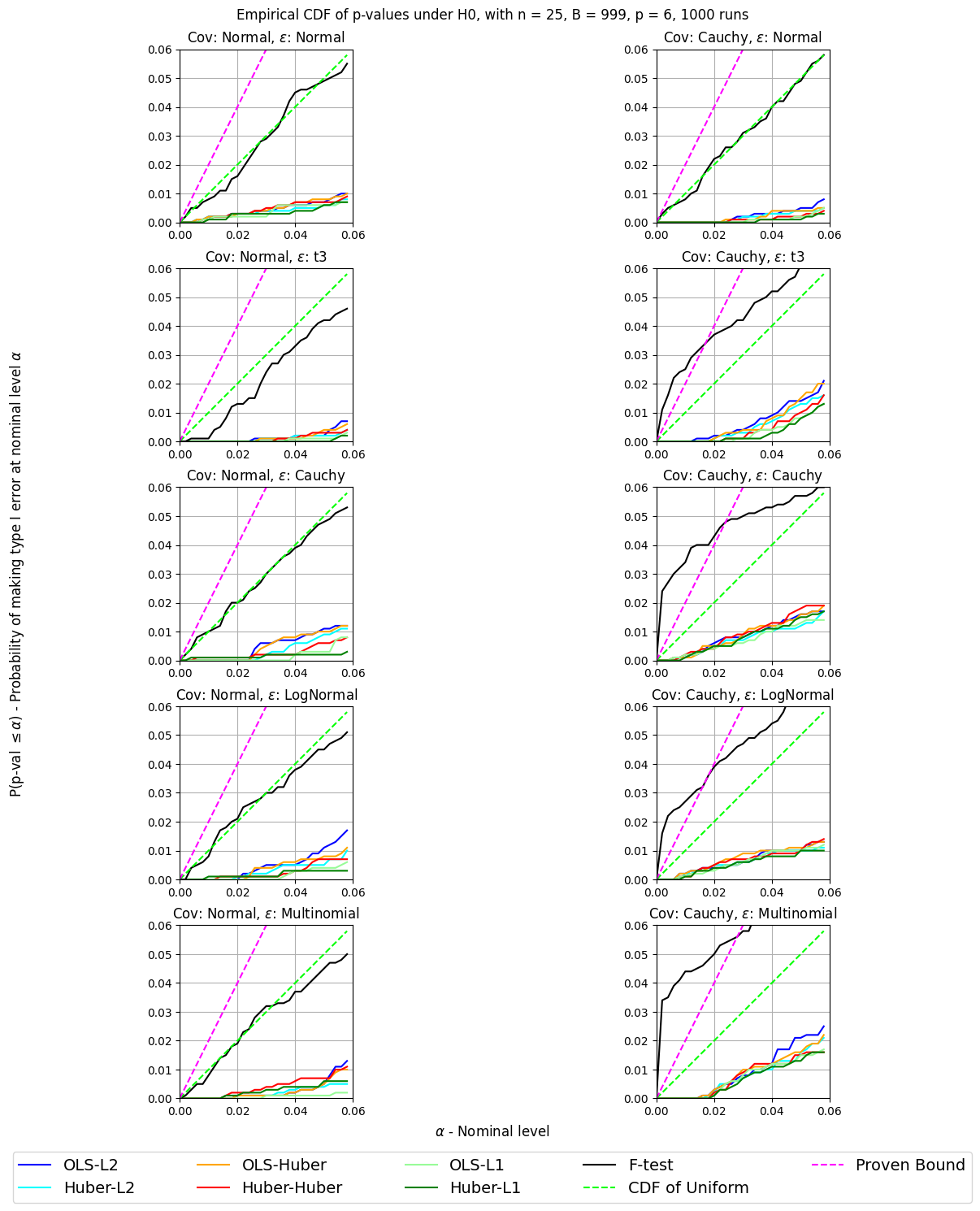}
    \vskip -0.25in
    \caption{
        Empirical CDF of RobustPALMRT and F-test p-values, with $\beta = 0$, $n = 25$, and $p = 6$.
        Ideally the empirical CDF would match the Uniform distribution CDF (the green dashed line) as closely as possible.
        We have proven that the CDF of the RobustPALMRT methods will lie below the $2\alpha$ line (the pink dashed line), but notice that empirically they fall at or below the Uniform distribution CDF.
    }
    \label{fig:n25-typeI-sim}
\end{figure}

\begin{figure}[!htbp]
    \centering
    \includegraphics[width=0.94\linewidth]{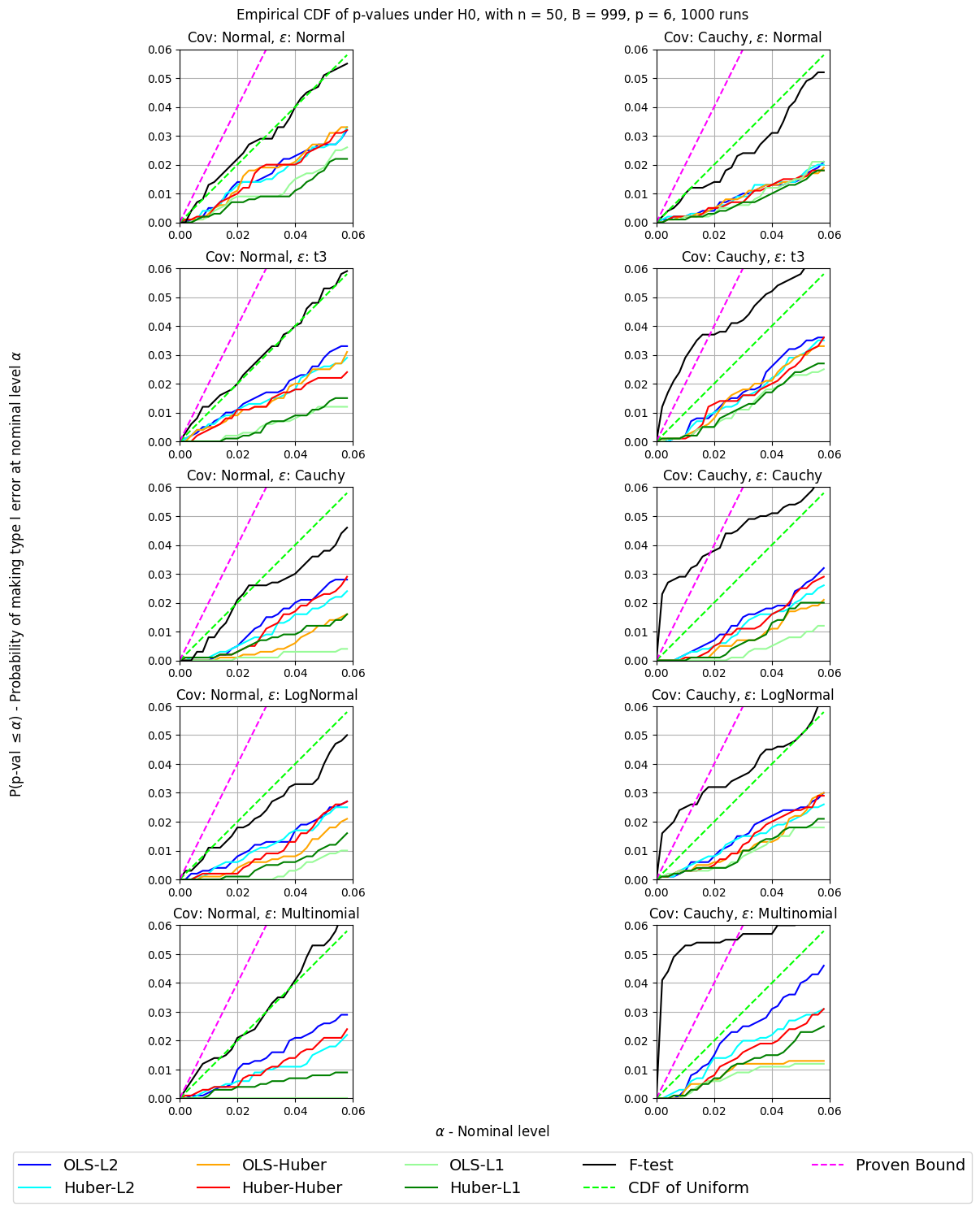}
    \vskip -0.25in
    \caption{
        Empirical CDF of RobustPALMRT and F-test p-values, with $\beta = 0$, $n = 50$, and $p = 6$.
        Ideally the empirical CDF (actual p-value) would match the Uniform CDF (nominal p-value; the green dashed line) as closely as possible.
        Our proofs ensure that the CDF of the RobustPALMRT methods lie below the $2\alpha$ line (the pink dashed line), but notice that empirically they fall at or below the Uniform CDF.
    }
    \label{fig:n50-typeI-sim}
\end{figure}

\begin{figure}[!htbp]
    \centering
    \includegraphics[width=0.94\linewidth]{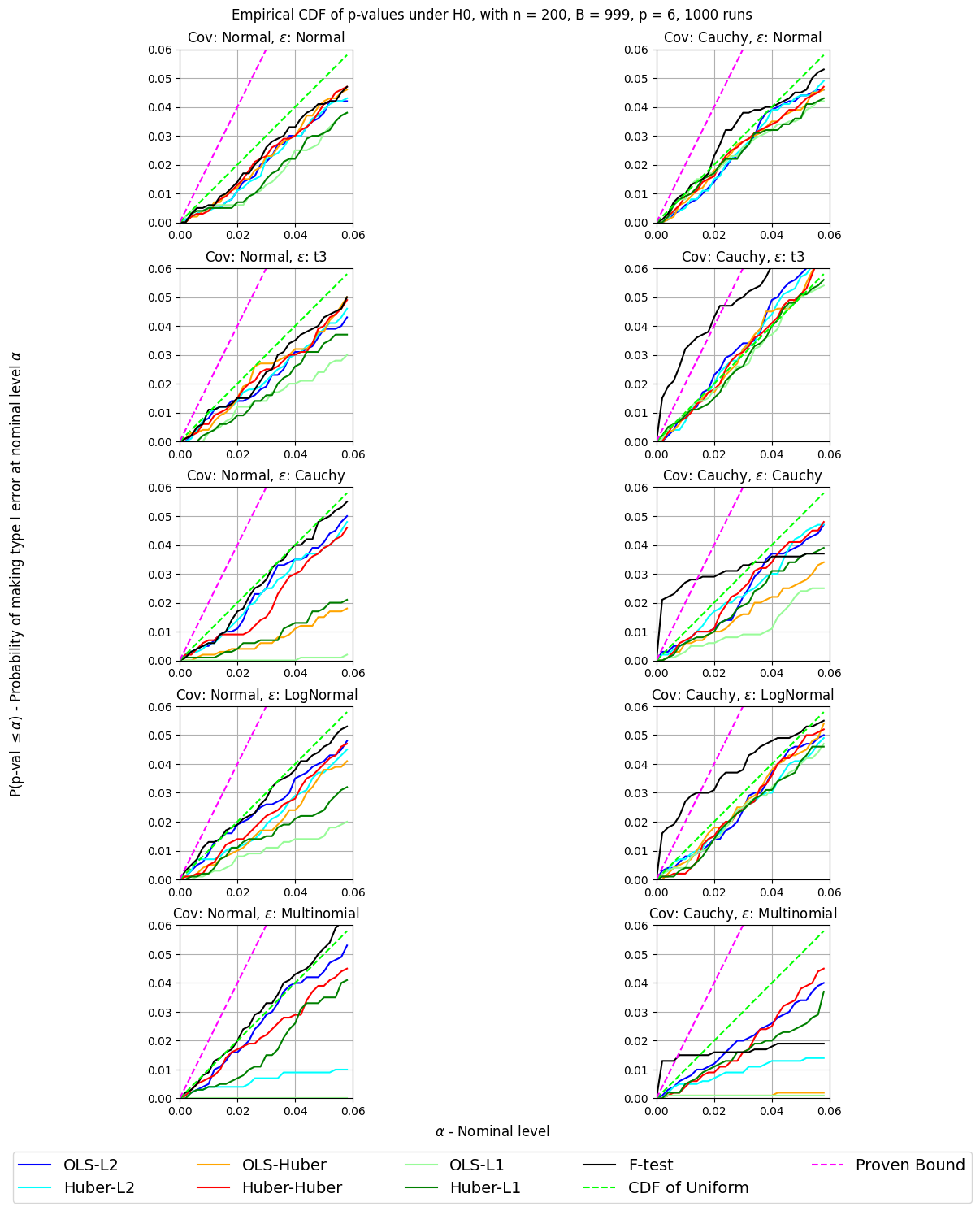}
    \vskip -0.25in
    \caption{
        Empirical CDF of RobustPALMRT and F-test p-values, with $\beta = 0$, $n = 200$, and $p = 6$.
        Ideally the empirical CDF (actual p-value) would match the Uniform CDF (nominal p-value; the green dashed line) as closely as possible.
        Our proofs ensure that the CDF of the RobustPALMRT methods lie below the $2\alpha$ line (the pink dashed line), but notice that empirically they fall at or below the Uniform CDF.
    }
    \label{fig:n200-typeI-sim}
\end{figure}

\begin{figure}[!htbp]
    \centering
    \includegraphics[width=0.94\linewidth]{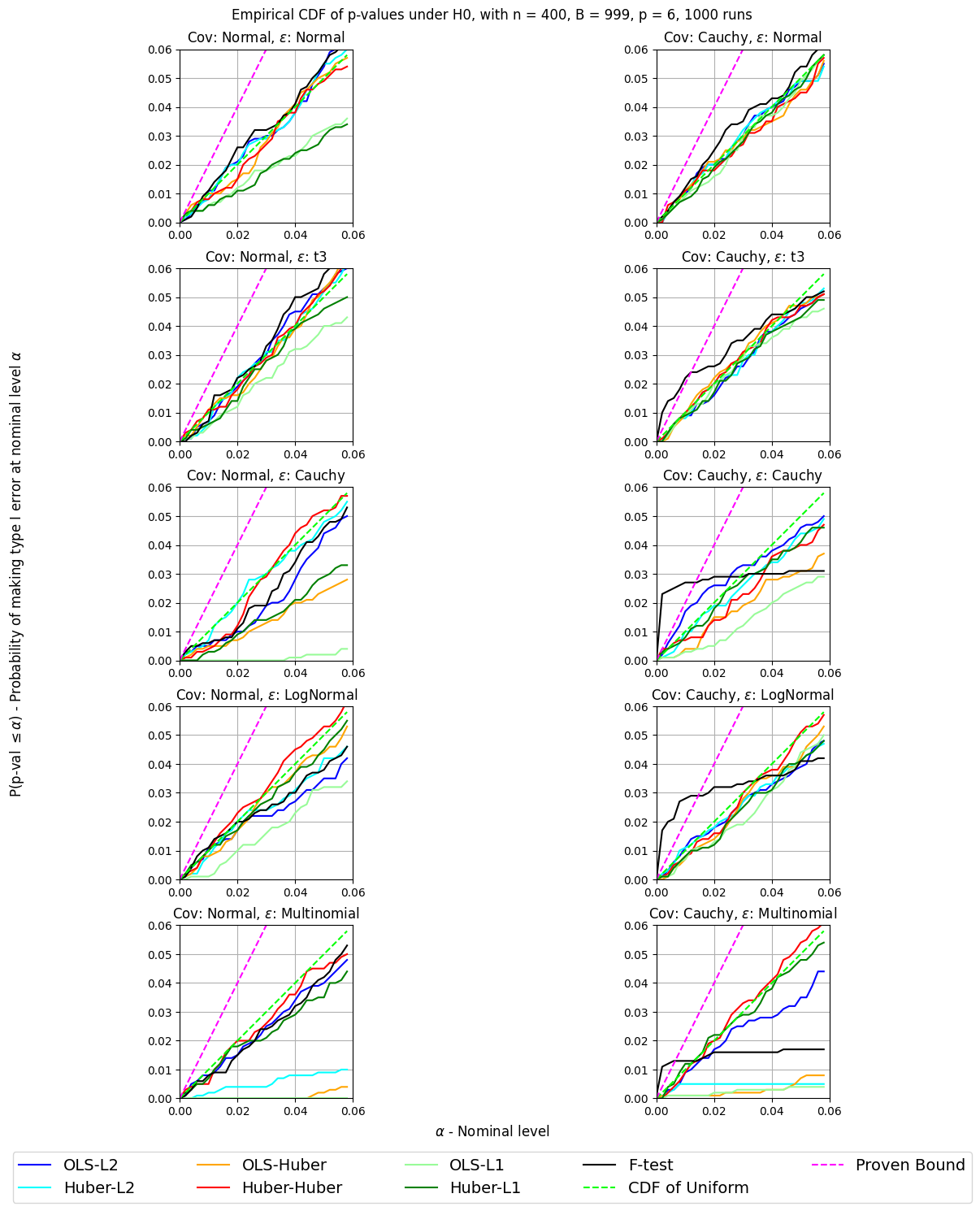}
    \vskip -0.25in
    \caption{
        Empirical CDF of RobustPALMRT and F-test p-values, with $\beta = 0$, $n = 400$, and $p = 6$.
        Ideally the empirical CDF (actual p-value) would match the Uniform CDF (nominal p-value; the green dashed line) as closely as possible.
        Our proofs ensure that the CDF of the RobustPALMRT methods lie below the $2\alpha$ line (the pink dashed line), but notice that empirically they fall at or below the Uniform CDF.
    }
    \label{fig:n400-typeI-sim}
\end{figure}

\end{document}